\documentclass[aap,MSNbibl,seceqn,dvips]{arximspdf}
\usepackage{accents}

\doi{10.1214/12-AAP887} 
\volume{23}
\issue{5}
\pubyear{2013}
\firstpage{1817}
\lastpage{1840}

\makeatletter

\renewcommand{\underbar}{\underaccent{\bar}}

\newproclaim{defn}{Definition}[section]
\newproclaim{rem}{Remark}[section]
\newtheorem{theorem}{Theorem}[section]
\newtheorem{lem}{Lemma}[section]
\newtheorem{prop}{Proposition}[section]
\newproclaim{asm}{Assumption}[section]

\renewcommand{\P}{\mathbb{P}}
\newcommand{\Q}{\mathbb{Q}}
\newcommand{\R}{\mathbb{R}}
\newcommand{\N}{\mathbb{N}}
\newcommand{\bS}{\mathbb{S}}
\newcommand{\F}{\mathcal{F}}
\newcommand{\A}{\mathcal{A}}
\newcommand{\B}{\mathcal{B}}
\newcommand{\M}{\mathcal{M}}
\newcommand{\V}{\mathcal{V}}
\newcommand{\wtV}{\tilde{V}}
\newcommand{\blamb}{\bar{\lambda}}
\newcommand{\C}{\mathcal{C}}
\newcommand{\eps}{\varepsilon}

\newcommand{\osc}{\operatorname{osc}}

\makeatother

\begin{document}
\begin{frontmatter}

\title{Robust maximization of asymptotic growth under covariance uncertainty}
\runtitle{Robust maximization of asymptotic growth}

\begin{aug}
\author[A]{\fnms{Erhan} \snm{Bayraktar}\corref{}\thanksref{t1}\ead[label=e1]{erhan@umich.edu}}
\and
\author[A]{\fnms{Yu-Jui} \snm{Huang}\ead[label=e2]{jayhuang@umich.edu}}
\runauthor{E. Bayraktar and Y.-J. Huang}
\affiliation{University of Michigan}
\address[A]{Department of Mathematics\\
University of Michigan\\
530 Church Street\\
Ann Arbor, Michigan 48109\\
USA\\
\printead{e1}\\
\hphantom{E-mail: }\printead*{e2}} 
\end{aug}

\thankstext{t1}{Supported in part by the National Science Foundation
under an applied
mathematics research Grant and a Career Grant, DMS-09-06257 and
DMS-09-55463, respectively,
and in part by the Susan M. Smith Professorship.}

\received{\smonth{7} \syear{2011}}
\revised{\smonth{7} \syear{2012}}

%
\begin{abstract}
This paper resolves a question proposed in Kardaras and Robertson
[\textit{Ann. Appl. Probab.} \textbf{22} (2012) 1576--1610]: how to invest in a
robust growth-optimal way in a market where precise knowledge of the
covariance structure of the underlying assets is unavailable. Among an
appropriate class of admissible covariance structures, we characterize
the optimal trading strategy in terms of a generalized version of the
principal eigenvalue of a fully nonlinear elliptic operator and its
associated eigenfunction, by slightly restricting the collection of
nondominated probability measures.
\end{abstract}

%

\end{frontmatter}

\section{Introduction}\label{sec1}
In this paper, we consider the problem of how to trade optimally in a
market when the investing horizon is long and the dynamics of the
underlying assets are uncertain. For the case where the uncertainty
lies only in the instantaneous expected return of the underlying
assets, this problem has been studied by Kardaras and Robertson
\cite{KR10}. They identify the optimal trading strategy using a
generalized version of the principle eigenfunction for a linear
elliptic operator which depends on the given covariance structure of
the underlying assets. We intend to generalize their results to the
case where \textit{even} the covariance structure of the underlying
assets is not known precisely, which is suggested in~\cite{KR10},
Discussion. More precisely, we would like to determine a
\textit{robust} trading strategy under which the asymptotic growth rate
of one's wealth, defined below, can be maximized no matter which
admissible covariance structure materializes.

Uncertainty in variance (or, equivalently, in covariance) has been
drawing increasing attention. The main difficulty lies in the absence
of one single dominating probability measure among $\Pi$, the
collection of all probability measures induced by variance uncertainty.
In their pioneering works, Avellaneda, Levy and Paras~\cite{ALP95} and
Lyons~\cite{LyonsTJ95} introduced the uncertain volatility model (UVM),
where the volatility process is only known to lie in a fixed interval
$[\underbar{\sigma},\bar{\sigma}]$. Under the Markovian
framework, they obtained a duality formula for the superhedging price
of (nonpath-dependent) European contingent claims. Under a generalized
version of the UVM, Denis and Martini~\cite{DM06} extended the above
duality formula, by using the capacity theory, to incorporate
path-dependent European contingent claims. For the capacity theory to
work, they required some continuity of the random variables being
hedged. Taking a different approach based on the underlying partial
differential equations, Peng~\cite{Peng07} derived results very similar
to~\cite{DM06}. The connection between~\cite{DM06} and~\cite{Peng07}
was then elaborated and extended in Denis, Hu and Peng~\cite{DHP10}. On
the other hand, instead of imposing some continuity assumptions on the
random variables being hedged, Soner, Touzi and Zhang \cite
{STZ11-aggregation} chose to restrict slightly the collection of
nondominated probability measures, and derived under this setting a
duality formulation for the superhedging problem. With all these
developments, superhedging under volatility uncertainty has then been
further studied in Nutz and Soner~\cite{NS10} and Nutz \cite
{Nutz12-Quasi-Sure}, among others. Also notice that Fernholz and
Karatzas~\cite{FK11} characterized the highest return relative to the
market portfolio under covariance uncertainty. Moreover, a
controller-and-stopper game with controlled drift and volatility was
considered in~\cite{BH11}, which can be viewed as an optimal stopping
problem under volatility uncertainty.

While we also take covariance uncertainty into account, we focus on
robust growth-optimal trading, which is different by nature from the
superhedging problem. Here, an investor intends to find a trading
strategy such that her wealth process can achieve maximal growth rate,
in certain sense, uniformly over all possible probability measures in
$\Pi$, or at least in a large enough subset $\Pi^*$ of $\Pi$. Previous
research on this problem can be found in~\cite{KR10} and the references
therein. It is worth noting that this problem falls under the umbrella
of ergodic control, for which the dynamic programming heuristic cannot
be directly applied; see, for example, Arapostathis, Borkar and Ghosh \cite
{ABG-book-11} and Borkar~\cite{Borkar06-ICM}, where they consider
ergodic control problems with controlled drift.

Following the framework in~\cite{KR10}, we first observe that the
associated differential operator under covariance uncertainty is a
variant of Pucci's extremal operator. We define the ``principal
eigenvalue'' for this fully nonlinear operator, denoted by $\lambda^*$,
in some appropriate sense, and then investigate the connection between
$\lambda^*$ and the generalized principal eigenvalue in~\cite{KR10}
where the covariance structure is a priori given. This connection is
first established on smooth bounded domains, thanks to the theory of
continuous selection in Michael~\cite{Michael56} and Brown \cite
{Brown89}. Next, observing that a Harnack inequality holds under
current context, we extend the result to unbounded domains. Finally, as
a consequence of this connection, we generalize~\cite{KR10}, Theorem
2.1, to the case with covariance uncertainty: we characterize the
largest possible asymptotic growth rate as $\lambda^*$ (which is robust
among probabilities in a large enough subset $\Pi^*$ of $\Pi$) and
identify the optimal trading strategy in terms of $\lambda^*$ and the
corresponding eigenfunction; see Theorem~\ref{thmmain}.

The structure of this paper is as follows. In Section~\ref{sec2}, we introduce
the framework of our study and formulate the problem of robust
maximization of asymptotic growth under covariance uncertainty. In
Section~\ref{sec3}, we first introduce several different notions of the
generalized principal eigenvalue and then investigate the relation
between them. The main technical result we obtain is Theorem \ref
{thmlamb*=inf}, using which we resolve the problem of robust
maximization of asymptotic growth in Theorem~\ref{thmmain}.

\subsection{Notation}\label{sec1.1} We collect some notation and definitions here for
readers' convenience:
\begin{itemize}
\item\mbox{$|\cdot|$} denotes the Euclidean norm in $\R^n$, and $\mathrm{Leb}$ denotes
Lebesgue measure in~$\R^n$.
\item$B_\delta(x)$ denotes the open ball in $\R^n$ centered at $x\in
\R^n$ with radius $\delta>0$.
\item$\bar{D}$ denotes the closure of $D$, and $\partial D$ denotes
the boundary of $D$.
\item Given $x\in\R^n$ and $D_1,D_2\subset\R^n$, $d(x,D_1):=\inf\{|x-y|
\mid y\in D_1\}$ and $d(D_1,D_2):=\inf\{|x-y| \mid x\in D_1, y\in D_2\}$.
\item Given $D\subset\R^n$, $C(D)=C^0(D)$ denotes the set of continuous
functions on $D$. If $D$ is open, $C^{k}(D)$ denotes the set of
functions having derivatives of order $\le k$ continuous in $D$, and
$C^{k}(\bar{D})$ denotes the set of functions in $C^k(D)$ whose
derivatives of order $\le k$ have continuous extension on $\bar{D}$.
\item Given $D\subset\R^n$, $C^{k,\beta}(D)$ denotes the set of
functions in $C^{k}(D)$ whose derivatives of order $\le k$ are H\"{o}lder
continuous on $D$ with exponent $\beta\in(0,1]$. Moreover,
$C_{\mathrm{loc}}^{k,\beta}(D)$ denotes the set of functions
belonging to $C^{k,\beta}(K)$ for every compact subset $K$ of $D$.
\item We say $D\subset\R^n$ is a domain if it is an open connected set.
We say $D$ is a smooth domain if it is a domain whose boundary is of
$C^{2,\beta}$ for some $\beta\in(0,1]$.
%
\item Given $D\subset\R^n$ and $u\dvtx D\mapsto\R$, $\mathop{\osc
}_{D}:= \sup\{|u(x)-u(y)| \mid x,y\in D\}$.
\end{itemize}

\section{The set-up}\label{sec2}
Fix $d\in\N$. Consider an open connected set $E\subseteq\mathbb{R}^d$,
and two functions $\theta,\Theta\dvtx  E\mapsto(0,\infty)$. The following
assumption will be in force throughout this paper.
%
\begin{asm}\label{asmstanding}
(i) $\theta$ and $\Theta$ are of $C_{
\mathrm{loc}}^{0,\alpha}(E)$ for some $\alpha\in(0,1]$, and $\theta<\Theta
$ in $E$.

(ii) There exists a sequence $\{E_n\}_{n\in\mathbb{N}}$ of
bounded open convex subsets of $E$ such that $\partial E_n$ is of
$C^{2,\alpha'}$ for some $\alpha'\in(0,1]$, $\bar{E}_n\subset E_{n+1}$
for all $n\in\mathbb{N}$ and $E=\bigcup_{n=1}^{\infty}E_n$.
\end{asm}

Let $\bS^d$ denote the space of $d\times d$ symmetric matrices,
equipped with the norm
%
\begin{equation}
\label{norm1} \|M\|:=\max_{i=1,\ldots, d}\bigl|e_i(M)\bigr|,\qquad  M\in
\bS^d,
\end{equation}
where $e_i(M)$'s are the eigenvalues of $M$. In some cases, we will
also consider the norm $\|M\|_{\max}:=\max|m_{ij}|$, for $M=\{m_{ij}\}_{i,j}\in\bS^d$. These two norms are equivalent with $\|\cdot\|_{\max
}\le\|\cdot\|\le d\|\cdot\|_{\max}$.
%
\begin{defn}\label{defnC}
Let $\C$ be the collection of functions $c\dvtx E \mapsto\bS^d$ such that:
\begin{longlist}[(ii)]
\item[(i)] for any $x\in E$, $\theta(x) |\xi|^2 \le\xi'c(x)\xi\le
\Theta(x) |\xi|^2, \forall \xi\in\R^d\setminus\{0\}$;
\item[(ii)] $c_{ij}(x)$ is of $C_{\mathrm{loc}}^{1,\alpha}(E)$,
$1\le i,j\le d$.
\end{longlist}
\end{defn}


Let $\widehat{E}:= E\cup\triangle$ be the one-point compactification of
$E$, where $\triangle$ is identified with $\partial E$ if $E$ is
bounded with $\partial E$ plus the point at infinity if $E$ is
unbounded. 
Following the set-up in~\cite{KR10}, Section 1, or~\cite{Pinsky-book}, page
40, we consider the space $C([0,\infty),\widehat{E})$ of
continuous functions $\omega\dvtx [0,\infty)\mapsto\widehat{E}$, and define
for each $\omega\in C([0,\infty),\widehat{E})$ the exit times
\[
\zeta_n(\omega):= \inf\{t\ge0 \mid\omega_t\notin
E_n\},\qquad \zeta (\omega):= \lim_{n\to\infty}\zeta_n(
\omega).
\]
Then, we introduce $\Omega:=  \{\omega\in C([0,\infty),\widehat
{E}) \mid \omega_{\zeta+t}=\triangle$ for all $t\ge
0$, if $\zeta(\omega)<\infty \}$.
Let $X=\{X_t\}_{t\ge0}$ be the coordinate mapping process for $\omega
\in\Omega$. Set $\{\mathcal{B}_t\}_{t\ge0}$ to be the natural
filtration generated by $X$, and denote by $\mathcal{B}$ the smallest
$\sigma$-algebra generated by $\bigcup_{t\ge0}\mathcal{B}_t$.
Similarly, set $(\mathcal{F}_t)_{t\ge0}$ to be the right-continuous
enlargement of $(\mathcal{B}_t)_{t\ge0}$, and denote by $\F$ the
smallest $\sigma$-algebra generated by $\bigcup_{t\ge0}\F_t$.
%
\begin{rem}
For financial applications, $X=\{X_t\}_{t\ge0}$ represents the
(relative) price process of certain underlying assets, and each $c\in
\mathcal{C}$ represents a possible covariance structure that might
eventually materialize. In view of Definition~\ref{defnC}(i), the
extent of the uncertainty in covariance is captured by the functions
$\theta$ and $\Theta$: they act as the pointwise lower and upper bounds
uniformly over all possible covariance structures $c\in\C$.
\end{rem}

\subsection{The generalized martingale problem}\label{sec2.1}
For any $M=\{m_{ij}\}_{i,j}\in\bS^d$, define the operator $L^M$ which
acts on $f\in C^2(E)$ by
\[
\bigl(L^Mf\bigr) (x):=\frac{1}{2}\sum
_{i,j=1}^{d}m_{ij}\,\frac{\partial^2
f}{\partial x_i\,\partial x_j}(x)=
\frac{1}{2}\operatorname{Tr}\bigl[M D^2f(x)\bigr],\qquad x\in E.
\]
For each $c\in\C$, we define similarly the operator $L^{c(\cdot)}$ as
\[
\bigl(L^{c(\cdot)}f\bigr) (x):=\frac{1}{2}\sum
_{i,j=1}^{d}c_{ij}(x)\,\frac
{\partial^2 f}{\partial x_i\,\partial x_j}(x)=
\frac{1}{2}\operatorname{Tr}\bigl[c(x) D^2f(x)\bigr],\qquad x\in E.
\]
Given $c\in\C$, a solution to the generalized martingale problem on $E$
for the operator $L^{c(\cdot)}$ is a family of probability measures
$(\Q^c_{x})_{x\in\widehat{E}}$ on $(\Omega,\mathcal{B})$ such that
$\Q^c_{x}[X_0=x]=1$ and
\[
f(X_{s\wedge\zeta_n})-\int_{0}^{s\wedge\zeta_n}
\bigl(L^{c(\cdot)} f\bigr) (X_u)\,du
\]
is a $(\Omega,(\mathcal{B}_t)_{t\ge0},\Q^c_{x})$-martingale for all
$n\in\N$ and $f\in C^2(E)$.

The following result, taken from~\cite{Pinsky-book}, Theorem 1.13.1,
states that Assumption~\ref{asmstanding} guarantees the existence and
uniqueness of the solutions to the generalized martingale problem on
$E$ for the operator $L^{c(\cdot)}$, for each fixed $c\in\C$.
%
\begin{prop}\label{martprob}
Under Assumption~\ref{asmstanding}, for each $c\in\C$, there is a
unique solution $(\Q^c_{x})_{x\in\widehat{E}}$ to the generalized
martingale problem on $E$ for the operator~$L^{c(\cdot)}$.
\end{prop}
%
\begin{rem}
For each $c\in\C$, as mentioned in~\cite{KR10}, Section 1,
\[
f(X_{s\wedge\zeta_n})-\int_{0}^{s\wedge\zeta_n}
\bigl(L^{c(\cdot)} f\bigr) (X_u)\,du
\]
is also a $(\Omega,(\mathcal{F}_t)_{t\ge0},\Q^c_{x})$-martingale for
all $n\in\N$ and $f\in C^2(E)$, as $f$ and $L^{c(\cdot)}f$ are bounded
in each $E_n$. Now, by taking $f(x)=x^i$, $i=1,\ldots,d$ and
$f(x)=x^ix^j$ with $i,j,=1\cdots d$, we get $X_{t\wedge\zeta_n}$ is a
$(\Omega,(\mathcal{F}_t)_{t\ge0},\Q^c_{x})$-martingale\vspace*{1pt}
with quadratic covariation process
$\int_0^\cdot1_{\{t\le\zeta_n\}}c(X_t)\,dt$, for each $n\in\N$ and
$x\in\widehat{E}$.
\end{rem}

\subsection{Asymptotic growth rate}\label{sec2.2}
For any fixed $x_0\in E$, we will simply write $\Q^c=\Q^c_{x_0}$ for
all $c\in\C$, when there is no confusion on the initial value $x_0$ of
$X$. Let us denote by $\Pi$ the collection of probability measures on
$(\Omega,\F)$ which are locally absolutely continuous with respect to
$\Q^c$ (written $\P\ll_{\mathrm{loc}}\Q^c$) for some $c\in
\C$,
and for which the process $X$ does not explode. That is,
\[
\Pi:=\bigl\{\P\in P(\Omega,\F) \mid\exists c\in\C \mbox{ s.t. } \P
|_{\F
_t}\ll\Q^c|_{\F_t} \mbox{ for all } t\ge0\mbox{,
and } \P[\zeta <\infty ]=0\bigr\},
\]
where $P(\Omega,\F)$ denotes the collection of all probability measures
on $(\Omega,\F)$. As observed in~\cite{KR10}, Section 1, for each
$\P
\in\Pi$, $X$ is a $(\Omega,(\mathcal{F}_t)_{t\ge0},\P
)$-semimar\-tingale such that $\P[X\in C([0,\infty),E)]=1$. Moreover, if
we take $c\in\C$ such that \mbox{$\P\ll_{\mathrm{loc}}\Q^c$}, then $X$
admits the representation
\[
X_\cdot= x_0+\int_0^\cdot
b^\P_t \,dt +\int_0^\cdot
\sigma (X_t)\,dW_t^\P,
\]
where\vspace*{1pt} $W^\P$ is a standard $d$-dimensional Brownian motion on $(\Omega
,(\mathcal{F}_t)_{t\ge0},\P)$, $\sigma$ is the unique symmetric
strictly positive definite square root of $c$, and $b^\P$ is a
$d$-dimensional $\{F_t\}_{t\ge0}$-progressively measurable process.

Let $(Z_t)_{t\ge0}$ be an adapted process. For $\P\in\Pi$, define
\[
\P\mbox{-}\liminf_{t\to\infty}Z_t:= \operatorname{ess}\operatorname{sup}^{\P}
\Bigl\{\chi \mbox{ is } \F\mbox{-measurable} \bigm| \lim_{t\to\infty}
\P[Z_t\ge \chi]=1 \Bigr\}.
\]
For any $d$-dimensional predictable process $\pi$ which is
$X$-integrable under $\Q^c$ for all $c\in\C$, we can define the process
$V^\pi_\cdot:=1+\int_{0}^{\cdot}\pi'_t\,dX_t$ under $\Q^c$ for all
\mbox{$c\in\C
$}. Let $\V$ denote the collection of all such processes $\pi$ which in
addition satisfy the following: for each $c\in\C$, $\Q^c[V^\pi_t>
0]=1,
\forall t\ge0$. Here, $\pi\in\V$ represents an admissible trading
strategy and $V^\pi$ represents the corresponding wealth process. Now,
for any $\pi\in\V$, we define the asymptotic growth rate of $V^\pi$
under $\P\in\Pi$ as
\[
g(\pi;\P):=\sup \Bigl\{\gamma\in\R \bigm| \P\mbox {-}\liminf_{t\to
\infty}
\bigl(t^{-1}\log V^\pi_t\bigr)\ge\gamma, \P
\mbox{-a.s.} \Bigr\}.
\]

\subsection{The problem}\label{sec2.3}\label{subsectheproblem}

The problem we consider in this paper is how to choose a trading
strategy $\pi^*\in \V$ such that the wealth process $V^{\pi^*}$ attains
the robust maximal asymptotic growth rate under all possible
probabilities in $\Pi$, or at least, in a large enough subset of $\Pi$
which readily contains all ``nonpathological'' cases. More precisely,
in Theorem~\ref{thmmain} below, we will construct a large enough
suitable subset $\Pi^*$ of $\Pi $, and determine
\[
\sup_{\pi\in\V}\inf_{\P\in\Pi^*}g(\pi;\P),
\]
the robust maximal asymptotic growth rate (robust in $\Pi^*$).
Moreover, we will find $\pi^*\in\V$ such that $V^{\pi^*}$ attains (or
surpasses) the maximal growth rate no matter which $\P\in\Pi^*$
materializes. This generalizes~\cite{KR10}, Theorem~2.1, to the case
with covariance uncertainty.

\section{The min--max result}\label{sec3}

In this section, we will first introduce generalized versions of the
principal eigenvalue for the linear operator $L^{c(\cdot)}$ and a fully
nonlinear operator $F$ defined below. Then, we will investigate the
relation between them on smooth bounded domains, and eventually extend
the result to the entire domain~$E$. The main technical result we
obtain is Theorem~\ref{thmlamb*=inf}. Finally, by using
Theorem~\ref{thmlamb*=inf}, we are able to resolve in
Theorem~\ref{thmmain} the problem proposed in
Section~\ref{subsectheproblem}.

%
%
%
%
%
%
%
%

Let us first recall the definition of Pucci's extremal operators. Given
$0<\lambda\le\Lambda$, we define for any $M\in\mathbb{S}^d$ the
following matrix operators:
%
\begin{eqnarray}
\label{defnPucci} \M^+_{\lambda,\Lambda}(M)&:=&\Lambda\sum
_{e_i(M)>0}e_i(M) +\lambda \sum
_{e_i(M)<0}e_i(M),
\nonumber\\[-8pt]\\[-8pt]
\M^-_{\lambda,\Lambda}(M)&:=&\lambda\sum_{e_i(M)>0}e_i(M)
+\Lambda \sum_{e_i(M)<0}e_i(M).\nonumber
\end{eqnarray}
From~\cite{CC-book}, page 15, we see that these operators can be
expressed as
\[
\M^+_{\lambda,\Lambda}(M)=\sup_{A\in\A(\lambda,\Lambda)}\operatorname{Tr}(AM),\qquad
\M^-_{\lambda,\Lambda}(M)=\inf_{A\in\A(\lambda,\Lambda
)}\operatorname{Tr}(AM),
\]
where $\A(a,b)$ denotes the set of matrices in $\bS^d$ with eigenvalues
lying in $[a,b]$ for some real numbers $a\le b$. For general properties
of Pucci's extremal operators, see, for example,~\cite{Pucci66} and
\cite{CC-book}, Section 2.2. Now, let us define the operator
$F\dvtx E\times
\mathbb{S}^d\mapsto\R$ by
%
\begin{equation}
\label{defnF} F(x,M):=\frac{1}{2}\M^+_{\theta(x),\Theta(x)}(M)=\frac{1}{2}
\sup_{A\in\A
(\theta(x),\Theta(x))}\operatorname{Tr}(AM).
\end{equation}

Let $D$ be an open connected subset of $E$. Fixing $c\in\C$, we
consider, for any given $\lambda\in\mathbb{R}$, the cone of positive
harmonic functions with respect to $L^{c(\cdot)}+\lambda$ as
%
\begin{equation}
H^c_\lambda(D):= \bigl\{\eta\in C^2(D) \mid
L^{c(\cdot)}\eta+\lambda \eta= 0 \mbox{ and } \eta>0 \mbox{ in } D \bigr\}
\end{equation}
and set
%
\begin{equation}
\lambda^{*,c}(D):= \sup\bigl\{\lambda\in\mathbb{R} \mid
H^c_\lambda (D)\neq \varnothing\bigr\}.
\end{equation}
Note that if $D$ is a smooth bounded domain, $\lambda^{*,c}(D)$
coincides with the principal eigenvalue for $L^{c(\cdot)}$ on $D$; see,
for example,~\cite{Pinsky-book}, Theorem 4.3.2. In our case, since we
do not require the boundedness of $D$, $\lambda^{*,c}(D)$ is a
generalized version of the principal eigenvalue for $L^{c(\cdot)}$ on
$D$, which is also used in~\cite{KR10}. On the other hand, for any
$\lambda\in\R$, we define
%
\begin{equation}
\label{defHlamb} H_\lambda(D):= \bigl\{\eta\in C^2(D) \mid F
\bigl(x,D^2\eta\bigr) + \lambda\eta \le0 \mbox{ and } \eta>0
\mbox{ in } D\bigr\}
\end{equation}
and set
%
\begin{equation}
\label{deflamb*} \lambda^*(D):=\sup\bigl\{\lambda\in\R\mid H_\lambda(D)
\neq\varnothing\bigr\},
\end{equation}
which is a generalized version of the principal eigenvalue for the
fully nonlinear operator $F$ on $D$. For auxiliary purposes, we also
consider, for any $\lambda\in\R$, the set
%
\begin{equation}
\label{defHlamb+} H^+_{\lambda}(D):= \bigl\{\eta\in C(\bar{D})\mid F
\bigl(x,D^2\eta\bigr) + \lambda \eta \le0 \mbox{ and } \eta>0
\mbox{ in } D\bigr\},
\end{equation}
where the inequality holds in the viscosity sense. From this, we define
%
\begin{equation}
\label{deflamb+} \lambda^{+}(D):=\sup\bigl\{\lambda\in\R\mid
H^+_{\lambda}(D)\neq \varnothing\bigr\}.
\end{equation}
For the special case where $D$ is a smooth bounded domain, $\lambda^+(D)$ is the principal half-eigenvalue of the operator $F$ on $D$ that
corresponds to positive eigenfunctions; see, for example,~\cite{QS08}.
%
\begin{lem}\label{lemexistence}
Given a smooth bounded domain $D\subset E$, there exists $\eta_{D}\in
C(\bar{D})$ such that $\eta_{D}>0$ in $D$ and satisfies in the
viscosity sense the equation
%
\begin{equation}
\label{Dirichlet} \cases{ F\bigl(x,D^2\eta_{D}\bigr) +
\lambda^+(D)\eta_{D}=0, &\quad in $D$,
\cr
\eta_{D}=0, &\quad on $\partial D$.}
\end{equation}
Moreover, for any pair $(\lambda, \eta)\in\R\times C(\bar{D})$ with
$\eta>0$ in $D$ which solves
%
\begin{equation}
\label{Dirichlet2} \cases{ F\bigl(x,D^2\eta\bigr) + \lambda\eta=0,
&\quad
in $D$,
\cr
\eta=0, &\quad on $\partial D$,}
\end{equation}
$(\lambda, \eta)$ must be of the form $ (\lambda^+(D),\mu\eta_{D} )$ for some $\mu>0$.
\end{lem}
\begin{pf}
Let us introduce some properties of $F$. By definition, we see that
%
\begin{eqnarray}
\label{H1}
&\displaystyle F(x,\mu M)=\mu F(x,M) \qquad\mbox{for any } x\in E \mbox{ and } \mu\ge
0;&
\\
\label{H2}
&\displaystyle  F \mbox{ is convex in } M.&
\end{eqnarray}
Also, by~\cite{CC-book}, Lemma 2.10(5), for any $x\in E$ and $M, N\in
\mathbb{S}^d$, we have
%
\begin{equation}
\label{H3} \tfrac{1}{2}\M^-_{\theta(x),\Theta(x)}(M-N) \le F(x,M)-F(x,N)\le
\tfrac
{1}{2}\M^+_{\theta(x),\Theta(x)}(M-N).\hspace*{-28pt}
\end{equation}
Finally, we observe from (\ref{defnPucci}) that $F$ can be expressed as
\[
F(x,M)=\frac{1}{2}\M^+_{\theta(x),\Theta(x)}(M)=\frac{1}{2} \biggl\{
\Theta (x)\sum_{e_i(M)>0}e_i(M)+\theta(x)\sum
_{e_i(M)<0}e_i(M) \biggr\}.
\]
From the continuity of $\theta$ and $\Theta$ in $x$, and the continuity
of $e_i(M)$ in $M$ for each $i$ (see, e.g.,~\cite{Meyer-book-00}, page
497), we conclude that
%
\begin{equation}
\label{H4} F \mbox{ is continuous in } E\times\mathbb{S}^d.
\end{equation}
Now, thanks to (\ref{H1})--(\ref{H4}) and~\cite{QS08}, Lemma 1.1, this
lemma follows from~\cite{QS08}, Theorems 1.1, 1.2.
\end{pf}

\subsection{\texorpdfstring{Regularity of $\eta_{D}$}
{Regularity of eta D}}\label{sec3.1}

In this subsection, we will
show that, for any smooth bounded domain $D\subset E$, the continuous
viscosity solution $\eta_{D}$ given in Lemma~\ref{lemexistence} is
actually smooth up to the boundary $\partial D$.\vadjust{\goodbreak} 

Let us consider the operator $J\dvtx \bar{D}\times\bS^d\mapsto\R$
defined by
\[
J(x,M):= F(x,M)+\lambda^+(D)\eta_{D}(x).
\]

\begin{lem}\label{lemHolder}
$\eta_{D}$ belongs to $C^{0,\beta}(\bar{D})$, for any $\beta\in(0,1)$.
\end{lem}
\begin{pf}
For any $x\in\bar{D}$ and $M, N\in\bS^d$ with $M\ge N$, we deduce from
(\ref{H3}) and (\ref{defnPucci}) that
%
\begin{eqnarray}
\label{H5} \frac{\theta_D}{2}\operatorname{Tr}(M-N)&\le&\frac{\theta(x)}{2}\operatorname{Tr}(M-N)=
\frac{1}{2}\M^-_{\theta(x),\Theta(x)}(M-N)\nonumber\\
&\le& F(x,M)-F(x,N)
\le\frac{1}{2}\M^+_{\theta(x),\Theta(x)}(M-N)\\
&=&\frac{\Theta
(x)}{2}\operatorname{Tr}(M-N)
\le\frac{\Theta_D}{2}\operatorname{Tr}(M-N),\nonumber
\end{eqnarray}
where $\theta_D:=\min_{x\in\bar{D}}\theta(x)$ and $\Theta_D:=\max_{x\in
\bar{D}}\Theta(x)$. On the other hand, recall that under
Assumption \ref
{asmstanding}, $\theta,\Theta\in C^{0,\alpha}(\bar{D})$. Let $K$
be a
H\"{o}lder constant for both $\theta$ and $\Theta$ on
$\bar
{D}$. By (\ref{defnF}) and (\ref{defnPucci}), for any $x,y\in\bar
{D}$ and $M\in\bS^d$,
%
\begin{eqnarray}
\label{H6}\qquad
&&
\bigl|F(x,M)-F(y,M)\bigr|\nonumber\\
&&\qquad\le\frac{1}{2} \biggl\{\bigl|\Theta(x)-\Theta(y)\bigr|\sum
_{e_i(M)>0}e_i(M)+\bigl|\theta(x)-\theta(y)\bigr|\sum
_{e_i(M)<0}\bigl|e_i(M)\bigr| \biggr\}
\\
&&\qquad\le Kd\|M\||x-y|^\alpha.\nonumber
\end{eqnarray}
Under (\ref{H1}), (\ref{H5}) and (\ref{H6}),~\cite{BD07}, Proposition
6, states that every bounded nonnegative viscosity solution to
%
\begin{equation}
\label{Jeqn} J\bigl(x,D^2\eta\bigr)=0 \mbox{ in } D,\qquad
\eta=0 \mbox{ on } \partial D
\end{equation}
is of the class $C^{0,\beta}(\bar{D})$ for all $\beta\in(0,1)$. Thanks
to Lemma~\ref{lemexistence}, $\eta_{D}$ is indeed a bounded
nonnegative viscosity solution to the above equation, and thus the
lemma follows.
\end{pf}
%
\begin{lem}\label{lemetaDsolvesJ}
$\eta_{D}$ is the unique continuous viscosity solution to (\ref{Jeqn}).
\end{lem}
\begin{pf}
By Lemma~\ref{lemexistence}, we immediately have the viscosity
solution property. To prove the uniqueness, it suffices to show that a
comparison principle holds for $J(x,D^2\eta)=0$. For any $x\in\bar{D}$
and $M, N\in\bS^d$ with $M\ge N$, we see from the definition of $J$ and
(\ref{H5}) that
%
\begin{equation}
\label{H5forJ} \frac{\theta_D}{2}\operatorname{Tr}(M-N)\le J(x,M)-J(x,N)\le
\frac{\Theta_D}{2}\operatorname{Tr}(M-N).
\end{equation}
Thanks to this inequality, we conclude from~\cite{KK07}, Theorem 2.6,
that a comparison principle holds for $J(x,D^2\eta)=0$.
\end{pf}

The following regularity result is taken from~\cite{Safonov88}, Theorem 1.2.
%
\begin{lem}\label{lemregularity}
Suppose $H\dvtx D\times\bS^d\mapsto\R$ satisfies the following conditions:
\begin{longlist}[(a)]
\item[(a)] $H$ is lower convex in $M\in\bS^d$;
\item[(b)] there is a $\nu\in(0,1]$ s.t. $\nu|\xi|^2 \le H(x,M+\xi
\xi')-H(x,M) \le\nu^{-1}|\xi|^2$ for all $\xi\in\R^d$;
\item[(c)] there is a $K_1>0$ s.t. $|H(x,0)|\le K_1$ for all $x\in D$;
\item[(d)] there are $K_2, K_3>0$ and $\beta\in(0,1)$ s.t. $\langle
H(\cdot,M)\rangle^{(\beta)}_{D}\le K_2 \sum_{i,j}|m_{ij}| +K_3$ for
all $M=\{m_{ij}\}_{i,j}\in\bS^d$, where
$\langle u\rangle^{(\beta)}_{D}:=\mathop{\sup}_{x\in D,
\rho>0} \rho^{-\beta}\mathop{\osc}_{D\cap B_{\rho}(x)} u$,
for any $u\dvtx D\mapsto\R$.
\end{longlist}
Then
\[
H\bigl(x,D^2\eta\bigr)=0 \mbox{ in } D,\qquad \eta=0 \mbox{ on }
\partial D
\]
has a unique solution in the class $C^{2,\beta}(\bar{D})$ if $\beta
\in
(0,\bar{\alpha})$, where the constant $\bar{\alpha}\in(0,1)$ depends
only on $d$ and $\nu$.
\end{lem}
%
\begin{prop}\label{propexistence}
$\eta_{D}$ belongs to $C^{2,\beta}(\bar{D})$ for any $\beta\in
(0,\alpha\wedge\bar{\alpha})$, where $\bar{\alpha}$ is given in
Lemma~\ref{lemregularity}. This in particular implies $\lambda^+(D)=\lambda^*(D)$, and thus we have
%
\begin{equation}
\label{Dirichletsmooth} \cases{ F\bigl(x,D^2\eta_{D}
\bigr) + \lambda^*(D)\eta_{D}=0, &\quad in $D$,
\cr
\eta_{D}=0, &\quad
on $\partial D$.}
\end{equation}
\end{prop}
\begin{pf}
Let us show that the operator $J$ satisfies conditions (a)--(d) in
Lem\-ma~\ref{lemregularity}. It is obvious from (\ref{H2}) that $J$
satisfies (a). Since $\xi\xi'\ge0$ and
$\operatorname{Tr}(\xi\xi')=|\xi|^2$ for all $\xi\in\R^d$, we see from
(\ref{H5forJ}) that $J$ satisfies (b). By the continuity of $\eta_{D}$
on $\bar{D}$, (c) is also satisfied as
$|J(x,0)|=0+\lambda^+(D)\eta_{D}(x) \le
K_1:=\lambda^{+}(D)\max_{\bar{D}}\eta_{D}$. To prove (d), let us first
observe that: for any $\beta\in(0,1)$ and $u\in C^{0,\beta}(D)$ with a
H\"{o}lder\vspace*{-1pt} constant $K$, we have $\mathop{\osc}_{D\cap B_{\rho}(x)}
u\le K\rho^\beta$, which yields $\langle u\rangle^{(\beta)}_{D}\le K$.
Recall that $\theta,\Theta\in C^{0,\alpha}(D)$ (Assumption
\ref{asmstanding}) and $\eta_{D}\in C^{0,\beta}(\bar{D})$ for all
$\beta\in(0,1)$ (Lemma~\ref{lemHolder}). Now, for any
$\beta\in(0,\alpha\wedge\bar{\alpha})$, we have $\theta, \Theta,
\eta_{D}\in C^{0,\beta}(D)$. Let $K'$ be a H\"{o}lder constant for all
the three functions. Then, from the definition of $J$, the calculation
(\ref{H6}) and the fact that $\|M\|\le d\|M\|_{\max}\le
d\sum_{i,j}|m_{ij}|$ for any $M=\{m_{ij}\}_{i,j}\in\bS^d$, we conclude
that $J(\cdot,M)\in C^{0,\beta}(D)$ with a H\"{o}lder constant $d^2
(\sum_{i,j}|m_{ij}| )K'+\lambda^+(D)K'$. It follows that $\langle
J(\cdot,M)\rangle^{(\beta)}_{D}\le d^2 (\sum_{i,j}|m_{ij}|
)K'+\lambda^+(D)K'$. Thus, (d) is satisfied for all
$\beta\in(0,\alpha\wedge\bar{\alpha})$, with $K_2:=d^2K'$ and
$K_3:=\lambda^+(D)K'$. Now, we conclude from Lemma~\ref{lemregularity}
that there is a unique solution in $C^{2,\beta}(\bar{D})$ to
(\ref{Jeqn}) for all $\beta\in(0,\alpha\wedge\bar{\alpha})$. However,
in view of Lemma~\ref{lemetaDsolvesJ}, this unique
$C^{2,\beta}(\bar{D})$ solution can only be $\eta_{D}$.

The fact that $\eta_{D}$ is of the class $C^{2,\beta}(\bar{D})$ and
solves (\ref{Dirichlet}) implies that $\lambda^{+}(D)\le\lambda^*(D)$.
Since we have the opposite inequality just from the definitions of
$\lambda^+(D)$ and $\lambda^*(D)$, we conclude that
$\lambda^{+}(D)=\lambda^*(D)$. Then (\ref{Dirichlet}) becomes (\ref
{Dirichletsmooth}).
\end{pf}


\subsection{\texorpdfstring{Relation between $\lambda^*(D)$ and $\lambda^{*,c}(D)$}
{Relation between lambda*(D) and lambda*,c(D)}}\label{sec3.2}

In this subsection, we will show that $\lambda^*(D) = \inf_{c\in\C}
\lambda^{*,c}(D)$ for any smooth bounded domain $D$.

Let us first state a maximum principle on small domains for the
operator $G_\delta\dvtx E\times\R\times\bS^d\mapsto\R$ defined by
\[
G_\delta(x,u,M):= -F(x,-M) -\delta|u|=\tfrac{1}{2}
\M^-_{\theta
(x),\Theta
(x)}(M)-\delta|u|,
\]
where $\delta$ can be any nonnegative real number.
%
\begin{lem}\label{lemmaxprinciple}
For any smooth bounded domain $D\subset E$, there exists \mbox{$\eps_0>0$},
depending on $D$, such that if a smooth bounded domain $U\subset D$
satisfies \mbox{$\operatorname{Leb}(U)<\eps_0$}, then if $\eta\in C(\bar{U})$
is a
viscosity solution to
\[
\cases{ G_\delta\bigl(x,\eta,D^2 \eta\bigr)\le0, &\quad in $U$,
\cr
\eta\ge0, &\quad on $\partial U$, }
\]
then $\eta\ge0$ in $U$.
\end{lem}
\begin{pf}
Consider the operator $\bar{F}\dvtx E\times\R\times\bS^d\mapsto\R$ defined
by $\bar{F}(x,u,\break M):=F(x,M)+\delta|u|$. For any $x\in E$, $u,v\in\R$ and
$M,N\in\bS^d$, we see from (\ref{H3}) that
%
\begin{eqnarray}
\label{H3barF} \qquad\tfrac{1}{2}\M^-_{\theta(x),\Theta(x)}(M-N)-\delta|u-v| &\le&\bar
{F}(x,u,M)-\bar{F}(x,v,N)\nonumber\\[-8pt]\\[-8pt]
&\le&\tfrac{1}{2}\M^+_{\theta(x),\Theta
(x)}(M-N)+
\delta|u-v|.\nonumber
\end{eqnarray}
Moreover, by (\ref{H4}), we immediately have
%
\begin{equation}
\label{H4barF} \bar{F}(x,0,M)=F(x,M) \mbox{ is continuous in } E\times
\bS^d.
\end{equation}
Noting that $G_\delta(x,u,M)=-\bar{F}(x,-u,-M)$, we have $G_\delta
(x,u,M)-G_\delta(x,v,\break N)=\bar{F}(x,-v,-N)-\bar{F}(x,-u,-M)$. Then, by
using (\ref{H3barF}), we get
%
\begin{eqnarray}
\label{DFforG} G_\delta(x,u-v,M-N)&=&\tfrac{1}{2}
\M^-_{\theta(x),\Theta
(x)}(M-N)-\delta |u-v|\nonumber\\
&\le& G_\delta(x,u,M)-G_\delta(x,v,N)
\nonumber\\[-8pt]\\[-8pt]
&\le&\tfrac{1}{2}\M^+_{\theta(x),\Theta(x)}(M-N)+\delta|u-v|\nonumber\\
&=&\bar {F}(x,u-v,M-N),\nonumber
\end{eqnarray}
which implies that the operator $G_\delta$ satisfies the $(D_F)$
condition in~\cite{QS08}, pa\-ge~107 (with $F$ replaced by $\bar{F}$).
Now, thanks to (\ref{H3barF})--(\ref{DFforG}), this lemma follows
from~\cite{QS08}, Theorem 3.5.
\end{pf}
%
\begin{prop}\label{proplamb*D<inflamb*cD}
For any smooth bounded domain $D\subset E$, $\lambda^*(D) \le\inf_{c\in
\C} \lambda^{*,c}(D)$.
\end{prop}
\begin{pf}
Assume the contrary that $\lambda^*(D)>\inf_{c\in\C}\lambda^{*,c}(D)$.
Then there exists $\bar{c}\in\C$ such that\vadjust{\goodbreak} $\lambda^*(D)>\lambda^{*,\bar
{c}}(D)$. Take $\bar{\eta}\in C^2(D)$ with $\bar{\eta}>0$ in $D$
such that
\[
\cases{ L^{\bar{c}(\cdot)}\bar{\eta}+\lambda^{*,\bar{c}}(D)\bar{\eta }=0, &\quad in
$D$,
\cr
\bar{\eta}=0, &\quad on $\partial D$.}
\]
From the definition of $F$, we see that $\bar{\eta}$ is a viscosity
subsolution to
%
\begin{equation}
\label{barceqn} F\bigl(x,D^2\eta\bigr)+\lambda^{*,\bar{c}}(D)
\eta=0 \qquad\mbox{in } D.
\end{equation}
On the other hand, the function $\eta_{D}$, given in Lemma \ref
{lemexistence}, is a viscosity supersolution to (\ref{barceqn}) as
it solves (\ref{Dirichletsmooth}) and $\lambda^*(D)>\lambda^{*,\bar
{c}}(D)$.
We claim that there exists $\ell>0$ such that $\bar{\eta}\le\ell
\eta_{D}$ in $D$. We will show this by following an argument used in
the proof of Theorem 4.1 in~\cite{QS08}. Take a compact subset $K$ of $D$
such that $\operatorname{Leb}(D\setminus K)<\eps_0$, where $\eps_0$ is given in
Lemma~\ref{lemmaxprinciple}. By the continuity of $\bar{\eta}$ and
$\eta_{D}$, there exists $\ell>0$ such that $\ell\eta_{D}-\bar
{\eta}>0$ on $K$. Consider the function $f_\ell:=\ell\eta_{D}-\bar
{\eta}$. By (\ref{H3}) and (\ref{H1}),
\begin{eqnarray*}
G_{\lambda^{*,\bar{c}}(D)}\bigl(x,f_\ell,D^2f_\ell
\bigr)&=&-F\bigl(x,-D^2f_{\ell
}\bigr)-\lambda^{*,\bar{c}}(D)|f_\ell|\\
&\le&-F\bigl(x,-D^2f_{\ell}\bigr)+\lambda^{*,\bar
{c}}(D)
f_\ell
\\
&\le&\ell F\bigl(x,D^2\eta_{D}\bigr)-F
\bigl(x,D^2\bar{\eta}\bigr)+\lambda^{*,\bar
{c}}(D) (\ell
\eta_{D}-\bar{\eta})\\
&\le&0 \qquad\mbox{in } D,
\end{eqnarray*}
where the last inequality follows from the supersolution property of
$\eta_{D}$ and the subsolution property of $\bar{\eta}$ to
(\ref{barceqn}). Since $f_\ell\ge0$ on $\partial(D\setminus K)$, we
obtain from Lemma~\ref{lemmaxprinciple} that $f_\ell\ge0$ on
$D\setminus K$. Thus, we conclude that $\bar{\eta}\le\ell\eta_{D}$
in $D$. Now, by Perron's method we can construct a continuous viscosity
solution $v$ to (\ref{barceqn}) on $D$ such that $\bar{\eta}\le
v\le
\ell\eta_{D}$. This in particular implies $v>0$ in $D$ and the pair
$(\lambda^{*,\bar{c}}(D),v)$ solves (\ref{Dirichlet2}). Recalling that
$\lambda^+(D)=\lambda^*(D)$ from Proposition~\ref{propexistence}, we
see that this is a contradiction to Lemma~\ref{lemexistence} as
$\lambda^{*,\bar{c}}(D)<\lambda^*(D)=\lambda^+(D)$.
\end{pf}

To prove the opposite inequality $\lambda^*(D)\ge\inf_{c\in\C
}\lambda^{*,c}(D)$ for any smooth bounded domain $D\subset E$, we will make use
of the theory of continuous selection pioneered by~\cite{Michael56},
and follow particularly the formulation in~\cite{Brown89}. For a brief
introduction to this theory and its adaptation to the current context,
see Appendix~\ref{secappendix}.
%
\begin{prop}\label{proplamb_d_inflambc_d}
Let $D\subset E$ be a smooth bounded domain. If $D$ is convex, then
$\lambda^*(D)\ge\inf_{c\in\C}\lambda^{*,c}(D)$.
\end{prop}
\begin{pf}
We will construct a sequence $\{\bar{c}'_m\}_{m\in\N}\subset\C$ such
that
\[
\limsup_{m\to\infty}\lambda^{*,\bar{c}'_m}(D)\le\lambda^*(D),
\]
which gives the desired result.

\textit{Step} 1: \textit{Constructing} $\{\bar{c}'_m\}_{m\in\N}$.
Recall that $\eta_{D}\in C^{2}(\bar{D})$ by Proposition~\ref
{propexistence}. Then, we deduce from (\ref{defnPucci}) that there
exists $\kappa>0$ such that
%
\begin{eqnarray}
\label{prop32-1}
&&\max\bigl\{\bigl|\lambda-\lambda'\bigr|,\bigl|\Lambda-
\Lambda'\bigr|\bigr\}<\kappa \nonumber\\[-8pt]\\[-8pt]
&&\quad\Rightarrow\quad \bigl|\M^+_{\lambda,\Lambda}
\bigl(D^2\eta_{D}(x)\bigr)-\M^+_{\lambda',\Lambda'}
\bigl(D^2\eta_{D}(x)\bigr)\bigr|<2/m\qquad\mbox{for all } x\in
\bar{D}.\hspace*{-30pt}\nonumber
\end{eqnarray}
Also,\vspace*{1pt} since \mbox{$\|\cdot\|_{\max}\le\|\cdot\|$}, the map $(M,x)\mapsto
L^M\eta_{D}(x)$ is continuous in $M$, uniformly in $x\in\bar{D}$.
It follows that there exists $\beta>0$ such that
%
\begin{equation}
\label{prop32-2} \|N-M\|<\beta \quad\Rightarrow\quad \bigl|L^N
\eta_{D}(x)-L^M\eta_{D}(x)\bigr|<1/m \qquad\mbox{for all }
x\in\bar{D}.\hspace*{-35pt}
\end{equation}
Set $\xi:=\min_{x\in\bar{D}}(\Theta-\theta)(x)>0$ (recall that
$\Theta
>\theta$ in $E$ under Assumption~\ref{asmstanding}). Now, by taking
$\gamma:=\theta+\frac{\kappa\wedge\xi}{4}$ and $\Gamma:=\Theta
-\frac
{\kappa\wedge\xi}{4}$ in Proposition~\ref{propcmexists}, we obtain
that there is a continuous function $c_m\dvtx \bar{D}\mapsto\bS^d$ such that
%
\begin{eqnarray}
\label{prop32-3} c_m(x)&\in&\A\bigl(\gamma(x),\Gamma(x)\bigr)
\quad\mbox{and}\nonumber\\[-8pt]\\[-8pt]
F_{\gamma,\Gamma
}\bigl(x,D^2\eta_{D}\bigr)&\le&
L^{c_m(\cdot)}\eta_{D}(x)+1/m \qquad\mbox{for all }
x\in\bar{D},\nonumber
\end{eqnarray}
where $F_{\gamma,\Gamma}(x,M)$ is defined in (\ref{defnFgamma}). By
mollifying the function $c_m$, we can construct a function $\bar
{c}_m\dvtx \bar{D}\mapsto\bS^d$ such that $\bar{c}_m\in C^\infty(\bar{D})$
and $\|\bar{c}_m(x)-c_m(x)\|_{\max}<(\beta\wedge\frac{\kappa
\wedge\xi
}{4})/d$ for all $x\in\bar{D}$ (more precisely, $c_m\in C(\bar{D})$
implies that for any open set $D'$ containing $\bar{D}$, there is a
function $\tilde{c}_m\in C(D')$ such that $\tilde{c}_m=c_m$ on $\bar
{D}$; see, e.g.,~\cite{GT-book-01}, Lemma 6.37. Then by mollifying
$\tilde
{c}_m$, we get a sequence of smooth functions converging uniformly to
$\tilde{c}_m$ on $\bar{D}$). It follows that
%
\begin{eqnarray}
\label{prop32-4} \bigl\|\bar{c}_m(x)-c_m(x)\bigr\|&\le& d\bigl\|
\bar{c}_m(x)-c_m(x)\bigr\|_{\max}\nonumber\\[-8pt]\\[-8pt]
&<&\beta \wedge
\frac{\kappa\wedge\xi}{4} \qquad\mbox{for all } x\in\bar{D}.\nonumber
\end{eqnarray}
Combining (\ref{prop32-1})--(\ref{prop32-4}), for each $x\in\bar
{D}$, we see that $\bar{c}_m(x)\in\A(\theta(x),\Theta(x))$ and
%
\begin{eqnarray}
\label{prop32-5}\qquad F\bigl(x,D^2\eta_{D}\bigr)&=&
\frac{1}{2}\M^+_{\theta(x),\Theta
(x)}\bigl(D^2\eta_{D}(x)
\bigr)<\frac{1}{2}\M^+_{\gamma(x),\Gamma(x)}\bigl(D^2
\eta_{D}(x)\bigr)+\frac
{1}{m}\nonumber\\
&=&F_{\gamma,\Gamma}
\bigl(x,D^2\eta_{D}\bigr)+\frac
{1}{m}
\le L^{c_m(\cdot)}\eta_{D}(x)+\frac{2}{m}\\
&\le&
L^{\bar
{c}_m(\cdot
)}\eta_{D}(x)+\frac{3}{m}.\nonumber
\end{eqnarray}
Now, take some $\bar{c}'_m\in\C$ such that $\bar{c}'_m$ and $\bar{c}_m$
coincide on $\bar{D}$. Then (\ref{prop32-5}) and the fact that
$F(x,D^2\eta_{D})+\lambda^*(D)\eta_{D}=0$ in $D$
(Proposition~\ref{propexistence}) imply
%
\begin{equation}
\label{hm} |h_m|<3/m \mbox{ in } D\qquad \mbox{where }
h_m:=L^{\bar{c}'_m(\cdot
)}\eta_{D}+\lambda^*(D)
\eta_{D}.
\end{equation}

\textit{Step} 2: \textit{Showing} $\limsup_{m\to\infty}\lambda^{*,\bar
{c}'_m}(D)\le\lambda^*(D)$. In the following, we will use the argument
in~\cite{Fujisaki99}, Section 3, starting from (3.3).
Let $\eta_m$ be the eigenfunction associated with the eigenvalue problem
\[
\cases{ L^{\bar{c}'_m(\cdot)}\eta+ \lambda^{*,\bar{c}'_m}(D)\eta=0, &\quad in $D$,
\cr
\eta=0, &\quad on $\partial D$.}
\]
Pick $x_0\in D$.\vspace*{1pt} We define the normalized eigenfunction $\tilde{\eta
}_m:=\frac{\eta_{D}(x_0)}{\eta_m(x_0)}\eta_m$. By~\cite{Pucci66b}, lemma on
page 789, there exist $k_1, k_2>0$, independent of $m$, such that
%
\begin{equation}
\label{etambdd} k_1 d(x,\partial D)\le\tilde{
\eta}_m(x)\le k_2 d(x,\partial D)\qquad \mbox{for all } x
\in D.
\end{equation}
Also, thanks to (\ref{H1}) and (\ref{H5}), we may apply
\cite{BD07}, Proposition 1, and obtain some $\delta>0$ and $C>0$ such that
$\eta_{D}(x)\le Cd(x,\partial D)$ if $d(x,\partial D)<\delta$.
Thus, we conclude that
%
\begin{equation}
\label{tmbdd} 1 \le t_m:=\sup_{x\in D}
\frac{\eta_{D}(x)}{\tilde{\eta}_m(x)}< \infty.
\end{equation}
By setting $s_m:=t_m\lambda^*(D)/\lambda^{*,\bar{c}'_m}(D)$, we deduce
from the definitions of $t_m$ and $s_m$ that
%
\begin{equation}
\label{L^c<0}\qquad L^{\bar{c}'_m(\cdot)}(s_m\tilde{
\eta}_m-\eta_{D})+ h_m=-t_m
\lambda^*(D)\tilde{\eta}_m+\lambda^*(D)\eta_{D}\le0
\qquad\mbox{in } D.
\end{equation}
Let $w_m$ be the unique solution of the class $C^{2,\alpha}(D)\cap
C(\bar{D})$ to the equation
%
\begin{equation}
\label{wmeqn} L^{\bar{c}'_m(\cdot)}w_m= h_m
\mbox{ in } D,\qquad w_m=0 \mbox{ on } \partial D.
\end{equation}
Note that by~\cite{Fujisaki99}, Remark 3.1, the convexity of $D$ and
(\ref{hm}) guarantee the existence of a constant $M>0$, independent of
$m$, such that
%
\begin{equation}
\label{wmbdd} \bigl|w_m(x)\bigr|\le\frac{M d(x,\partial D)}{m}
\qquad\mbox{for all } x\in D.
\end{equation}
Combining (\ref{L^c<0}) and (\ref{wmeqn}), we get
\[
\cases{ L^{\bar{c}'_m(\cdot)}(s_m\tilde{\eta}_m-
\eta_{D}+w_m) \le0, &\quad in $D$,
\cr
s_m\tilde{
\eta}_m-\eta_{D}+w_m=0, &\quad on $\partial D$. }
\]
We then conclude from the maximum principle that $s_m\tilde{\eta
}_m-\eta_{D}+w_m\ge0$ in $D$. From the definition of $s_m$, this
inequality gives
\[
\frac{\lambda^*(D)}{\lambda^{*,\bar{c}'_m}(D)}\ge\frac{\eta_{D}(x)}{t_m\tilde{\eta}_m(x)}-\frac{w_m(x)}{t_m\tilde{\eta
}_m(x)}\ge
\frac{\eta_{D}(x)}{t_m\tilde{\eta}_m(x)}-\frac{M}{k_1 m} \qquad\mbox{for all } x\in D,
\]
where the last inequality follows from (\ref{wmbdd}), (\ref{tmbdd})
and (\ref{etambdd}). Now, take a sequence $\{x_k\}_{k\in\N}$ in $D$
such that $\frac{\eta_{D}(x_k)}{\eta_m(x_k)}\to t_m$. By plugging
$x_k$ into the above inequality and taking limit in $k$, we get
\[
\frac{\lambda^*(D)}{\lambda^{*,\bar{c}'_m}(D)}\ge1-\frac{M}{k_1 m},
\]
which implies $\lambda^*(D)\ge\limsup_{m\to\infty}\lambda^{*,\bar
{c}'_m}(D)$.\vadjust{\goodbreak}
\end{pf}

Combining Propositions~\ref{proplamb*D<inflamb*cD} and \ref
{proplamb_d_inflambc_d}, we have the following result:
%
\begin{theorem}\label{thmlamb*D=inflamb*cD}
Let $D\subset E$ be a smooth bounded domain. If $D$ is convex, $\lambda^*(D)=\inf_{c\in\C}\lambda^{*,c}(D)$.
\end{theorem}

\subsection{\texorpdfstring{Relation between $\lambda^*(E)$ and $\lambda^{*,c}(E)$}
{Relation between lambda*(E) and lambda*,c(E)}}\label{sec3.3}

In this subsection, we will first characterize $\lambda^*(E)$ in terms
of $\lambda^*(E_n)$, and then generalize Theorem \ref
{thmlamb*D=inflamb*cD} from bounded domains to the entire space $E$.

Let us first consider some Harnack-type inequalities. Note that for any
$D\subset\R^d$ and $p\in[1,\infty)$, we will denote by $\mathcal
{L}^p(D)$ the space of measurable functions $f$ satisfying $(\int_D|f(x)|^p\,dx)^{1/p}<\infty$.
%
\begin{lem} \label{lem_harnack}
Let $D\subset E$ be a smooth bounded domain. Let $H\dvtx E\times\bS^d\mapsto\R$ be such that
%
\begin{eqnarray}
\label{Hbdd}
&&\exists 0<\lambda\le\Lambda \quad\mbox{s.t.}\quad
\M^-_{\lambda
,\Lambda
}(M)\le H(x,M)\le\M^+_{\lambda,\Lambda}(M) \nonumber\\[-8pt]\\[-8pt]
&&\eqntext{\displaystyle \mbox{for all }
(x,M)\in D\times\bS^d.}
\end{eqnarray}
If $\{u_n\}_{n\in\N}$ is sequence of continuous nonnegative viscosity
solutions to
%
\begin{equation}
\label{uneqn} H\bigl(x,D^2u_n\bigr)+
\delta_n u_n=f_n \qquad\mbox{in } D,
\end{equation}
where $\{\delta_n\}_{n\in\N}$ is a bounded sequence in $[0,\infty)$ and
$f_n\in\mathcal{L}^d(D)$, then we have:
\begin{longlist}[(ii)]
\item[(i)] for any compact set $K\subset D$, there is a constant
$C>0$, depending only on $D$, $K$, $d$, $\lambda$, $\Lambda$, $\sup_n\delta_n$, such that
%
\begin{equation}
\label{Harnack} \sup_{K}u_n\le C \Bigl\{
\inf_K u_n+\|f_n\|_{\mathcal{L}^d(D)} \Bigr\}.
\end{equation}
\item[(ii)] Suppose $H$ satisfies (\ref{H1}). Given $x_0\in D$ and
$R_0>0$ such that\break $B_{R_0}(x_0)\subset D$, there exists a constant
$C>0$, depending only on $R_0$, $d$, $\lambda$, $\Lambda$, $\sup_n\delta_n$, such that for any $0<R<R_0$,
%
\begin{equation}
\label{Harnack2} \sup_{\bar{B}_R(x_0)}u_n\le C \Bigl\{
\inf_{\bar{B}_R(x_0)} u_n+R^2\| f_n
\|_{\mathcal{L}^d(B_{R_0}(x_0))} \Bigr\}.
\end{equation}
As a consequence, if we assume further that $\{u_n\}_{n\in\N}$ is
uniformly bounded, and $\{f_n\}_{n\in\N}$ is bounded in $\mathcal
{L}^d(D)$, then for any compact connected set $K\subset D$ and $\beta
\in
(0,1)$, $u_n\in C^{0,\beta}(K)$ for all $n\in\N$, with one fixed
H\"{o}lder constant.
\end{longlist}
\end{lem}
\begin{pf}
(i) Set $\delta^*:=\sup_n\delta_n<\infty$. By (\ref{Hbdd}), we have
\[
\M^+_{\lambda,\Lambda}\bigl(D^2u_n\bigr)+
\delta^*u_n \ge H\bigl(x,D^2u_n\bigr)+
\delta_n u_n \ge\M^-_{\lambda,\Lambda}\bigl(D^2u_n
\bigr)-\delta^*u_n \qquad\mbox{in } D.
\]
In view\vspace*{1pt} of (\ref{uneqn}), we obtain $\M^+_{\lambda,\Lambda
}(D^2u_n)+\delta^*u_n \ge f_n \ge\M^-_{\lambda,\Lambda
}(D^2u_n)-\delta^*u_n$ in $D$. Thanks to this inequality, estimate (\ref{Harnack})
follows from~\cite{QS08}, Theorem~3.6.\vadjust{\goodbreak}

(ii) Thanks to estimate (\ref{Harnack}) and
\cite{GT-book-01}, Lemma 8.23, we can prove part~(ii) by following the argument in
the proof of Corollary 3.2 in~\cite{BD10}. For a detailed proof, see
Appendix~\ref{secappendixB}.
\end{pf}
%
\begin{prop}\label{proplamb*=limlamb*}
$\lambda^*(E)= \downarrow\lim_{n\to\infty}\lambda^*(E_n)$ and there
exists some $\eta^*\in H_{\lambda^*(E)}(E)$ such that
%
\begin{equation}
\label{eta*eqn} F\bigl(x,D^2\eta^*\bigr)+\lambda^*(E)\eta^*=0
\qquad\mbox{in } E.
\end{equation}
\end{prop}
\begin{pf}
It is obvious from the definition that $\lambda^*(E_n)$ is decreasing\break
in $n$ and $\lambda^*(E)\le\lambda^*(E_n)$ for all $n\in\N$. It follows
that $\lambda^*(E)\le\lambda_0:= \downarrow\lim_{n\to\infty
}\lambda^*(E_n)$. To prove the opposite inequality, it suffices to show that
$H_{\lambda_0}(E)\neq\varnothing$. To this end, we take $\eta_n$ as the
eigenfunction given in Lemma~\ref{lemexistence} with $D=E_n$. Pick an
arbitrary $x_0\in E_1$, and define $\tilde{\eta}_n(x):=\frac{\eta_n(x)}{\eta_n(x_0)}$ such that $\tilde{\eta}_n(x_0)=1$ for all $n\in
\N$.\vspace*{1pt}

Fix $n\in\N$. In view of Proposition~\ref{propexistence}, $\{\tilde
{\eta}_m\}_{m> n}$ is a sequence of positive smooth solutions to
%
\begin{equation}
\label{etameqn} F\bigl(x,D^2\tilde{
\eta}_m\bigr)+\lambda^*(E_m)\tilde{\eta}_m=0
\qquad\mbox{in } E_{n+1}.
\end{equation}
From the definition of $F$, we see that $F$ satisfies (\ref{Hbdd}) in
$E_n$ with $\lambda=\min_{x\in\bar{E}_n}\theta(x)$ and $\Lambda
=\max_{x\in\bar{E}_n}\Theta(x)$. Thus, by Lemma~\ref{lem_harnack}(i),
there is a constant $C>0$, independent of $m$, such that
\[
\sup_{\bar{E}_n}\tilde{\eta}_m\le C\inf_{\bar{E}_n}\tilde{
\eta }_m\le C,
\]
which implies $\{\tilde{\eta}_m\}_{m>n}$ is uniformly bounded in $\bar
{E}_n$. On the other hand, given $\beta\in(0,1)$, Lemma
\ref{lem_harnack}(ii) guarantees that $\tilde{\eta}_m\in C^{0,\beta
}(\bar {E}_n)$ for all $m>n$, with a fixed H\"{o}lder
constant. Therefore, by using the Arzela--Ascoli theorem, we conclude
that $\tilde{\eta}_m$ converges uniformly, up to some subsequence, to
some function $\eta^*$ on $\bar{E}_n$. Thanks to the stability result
of viscosity solutions (see, e.g.,~\cite{FS-book-06}, Lemma II.6.2), we
obtain from (\ref{etameqn}) that $\eta^*$ is a nonnegative continuous
viscosity solution in $E_n$ to
%
\begin{equation}
\label{eta*eqn2} F\bigl(x,D^2\eta^*\bigr)+\lambda_0
\eta^*=0.
\end{equation}
Furthermore, since $\eta^*(x_0)=\lim_{m\to\infty}\eta_m(x_0)=1$, we
conclude from~\cite{BD07}, Theorem 2, a strict maximum principle for
eigenvalue problems of fully nonlinear operators, that $\eta^*>0$ in
$E_n$. Finally, noting that for any $\beta\in(0,1)$, $\eta^*\in
C^{0,\beta}(\bar{E}_n)$ with its H\"{o}lder constant same
as $\tilde{\eta}_m$'s, we may use Lemma~\ref{lemregularity}, as in the
proof of Proposition~\ref{propexistence}, to show that $\eta^*\in
C^2(\bar{E}_n)$.

Since the results above hold for each $n\in\N$, we conclude that
$\eta^*$ belongs to $C^2(E)$, takes positive values in $E$ and satisfies
(\ref{eta*eqn2}) in $E$. It follows that $\eta^*\in H_{\lambda
_0}(E)$, which yields $\lambda_0\le\lambda^*(E)$. Therefore, we get
$\lambda^*(E)=\lambda_0$, and then (\ref{eta*eqn2}) becomes (\ref
{eta*eqn}).
\end{pf}

Now, we are ready to present the main technical result of this paper.
%
\begin{theorem}\label{thmlamb*=inf}
$\lambda^*(E) = \inf_{c\in\C} \lambda^{*,c}(E)$.
\end{theorem}
\begin{pf}
Thanks to Theorem 4.4.1(i) in~\cite{Pinsky-book}, Theorem
\ref{thmlamb*D=inflamb*cD} and Proposition~\ref{proplamb*=limlamb*},
we have
\[
\inf_{c\in\C}\lambda^{*,c}(E) = \inf_{c\in\C}
\inf_{n\in\N
}\lambda^{*,c}(E_n) = \inf_{n\in\N}
\inf_{c\in\C}\lambda^{*,c}(E_n) =\inf_{n\in\N
}
\lambda^*(E_n)=\lambda^*(E).\quad
\]
\upqed\end{pf}
%
\begin{rem}
For the special case where $\theta$ and $\Theta$ are merely two
positive constants, the derivation of Theorem~\ref{thmlamb*=inf} can
be much simpler. Since the operator $F(x,M)=\frac{1}{2}\M^+_{\theta
,\Theta}(M)$ is now Pucci's operator with elliptic constants $\theta$
and $\Theta$, we may apply~\cite{BD10}, Theorem 3.5, 
and obtain a positive H\"{o}lder continuous viscosity
solution $\eta^*$ to
\[
F\bigl(x,D^2\eta^*\bigr)+\bar{\lambda}(E)\eta^*=0
\qquad\mbox{in } E,
\]
where $\bar{\lambda}(E):=\inf\{\lambda^+(D) \mid D\subset E$ is a
smooth bounded domain$\}$. Then, Lem\-ma~\ref{lemregularity} implies
$\eta^*$ is actually smooth, and thus
$\bar{\lambda}(E)\le\lambda^*(E)$. Since
$\bar{\lambda}(E)\ge\lambda^*(E)$ by definition, we conclude that
$\bar{\lambda}(E)=\lambda^*(E)$. Now, thanks to~\cite{Pinsky-book},
Theorem~4.4.1(i), and the standard result $\lambda^+(E_n) =
\inf_{c\in\C}\lambda^{*,c}(E_n)$ for Pucci's operator (see, e.g.,
\cite{BEQ05}, Proposition 1.1(ii), and~\cite{Pucci66b}, Theorem I), we
get
\begin{eqnarray*}
\inf_{c\in\C}\lambda^{*,c}(E) &=& \inf_{c\in\C}
\inf_{n\in\N
}\lambda^{*,c}(E_n) = \inf_{n\in\N}
\inf_{c\in\C}\lambda^{*,c}(E_n) \\
&=& \inf_{n\in
\N}
\lambda^+(E_n)=\blamb(E)=\lambda^*(E).
\end{eqnarray*}

However, as pointed out in~\cite{KR10}, Discussion, it is not
reasonable for financial applications to assume that each $c\in\C$ is
both continuous and uniformly elliptic in~$E$. Therefore, we consider
in this paper the more general setting where $\theta$ and $\Theta$ are
functions defined on $E$, which includes the case without uniform ellipticity.
\end{rem}

\subsection{Application}\label{sec3.4}

By Theorem~\ref{thmlamb*=inf} and mimicking the proof of Theorem~2.1 in
\cite{KR10}, we have the following result. Note that, for simplicity,
we will write $\lambda^*=\lambda^*(E)$.
%
\begin{theorem}\label{thmmain}
Take $\eta^*\in H_{\lambda^*}(E)$ and normalize it so that $\eta^*(x_0)=1$. Define $\pi^*_t:= e^{\lambda^*t}\nabla\eta^*(X_t)$ for all
$t\ge0$, and set
\[
\Pi^*:= \Bigl\{\P\in\Pi \bigm| \P\mbox{-}\liminf_{t\to
\infty
}
\bigl(t^{-1}\log\eta^*(X_t)\bigr) \ge0, \P\mbox{-a.s.}
\Bigr\}.
\]
Then, we have $\pi^*\in\V$ and $g(\pi^*;\P)\ge\lambda^*$ for all
$\P\in
\Pi^*$. Moreover,
%
\begin{equation}
\label{min-maxproblem} \lambda^* = \sup_{\pi\in\V}\inf_{\P\in\Pi^*}g(
\pi;\P) = \inf_{\P\in\Pi
^*}\sup_{\pi\in\V}g(\pi;\P).\vadjust{\goodbreak}
\end{equation}
\end{theorem}
\begin{pf}
Set $V^*_t:=V^{\pi^*}_t=1+\int_0^t
e^{\lambda^*s}\nabla\eta^*(X_s)'\,dX_s$, $t\ge0$. By applying
It\^{o}'s rule to the process $e^{\lambda^*t}\eta^*(X_t)$ we
see that $V^*_t\ge e^{\lambda ^*t}\eta^*(X_t)>0$ $\P$-a.s. for all
\mbox{$\P\in\Pi$}. This already implies $\pi^*\in\V$. Also, by the
construction of $\Pi^*$, we have $\P
$-$\liminf_{t\to\infty}(t^{-1}\log(V^*_t))\ge\lambda^*$ $\P $-a.s. for
all $\P\in\Pi^*$. It follows that $g(\pi^*;\P)\ge\lambda^*$ for all $\P
\in\Pi^*$, which in turn implies $\lambda^* \le\break\sup_{\pi\in\V }\inf_{\P
\in\Pi^*}g(\pi;\P)$.\looseness=1

Now,\vspace*{-1pt} for any $c\in\C$ and $n\in\N$, set $\lambda^{*,c}_n
= \lambda^{*,c}(E_n)$, take $\eta^{*,c}_n\in
H^c_{\lambda^{*,c}_n}(E_n)$ with $\eta^{*,c}_n(x_0)=1$ and define the
process $\wtV^c_n(t):=e^{\lambda ^{*,c}_n t}\eta^{*,c}_n(X_t)$. Note
that under any $\P\in\Pi$ such that $\P\ll_{\mathrm{loc}}\Q^c$, we have
$\wtV^c_n(t) = 1+\int_{0}^{t}(\pi^{*,c}_n)'_s \,dX_s$ with
$(\pi^{*,c}_n)_{t}:=e^{\lambda^{*,c}_n t}\nabla \eta^{*,c}_n(X_t)$.
This, however, may not be true for general $\P\in \Pi $. As shown in
the proof of Theorem 2.1 in~\cite{KR10}, for any fixed $c\in\C$ and
$n\in\N$, we have the following: (1) there exists a solution
$(\P^{*,c}_{x,n})_{x\in E_n}$ to the generalized martingale problem for
the operator $L^{c(\cdot),\eta^{*,c}_n}:=L^{c(\cdot)} +
c\nabla\log\eta^{*,c}_n\cdot\nabla$; (2) the coordinate process $X$
under $(\P^{*,c}_{x,n})_{x\in E_n}$ is recurrent\vspace*{1pt} in $E_n$;
(3) $\P^{*,c}_{x,n}\ll_{\mathrm{loc}}\Q^c$ (note that we conclude from the
previous two conditions that $\P^{*,c}_{x,n}\in\Pi^*$); (4) the process
$V^{\pi}/\wtV^c_n$ is a nonnegative $\P^{*,c}_{x,n}$-supermartingale
for all $\pi\in\V$. We therefore have the analogous result
$g(\pi;\P^{*,c}_n)\le g(\pi^{*,c}_n;\P^{*,c}_n)\le \lambda^{*,c}_n$ for
all $\pi\in\V$, which yields $\inf_{\P\in
\Pi^*}\sup_{\pi\in\V}g(\pi;\P)\le\lambda^{*,c}_n$. Now, thanks to
Theorem~4.4.1(i) in~\cite{Pinsky-book} and Theorem~\ref{thmlamb*=inf},
we have
\[
\inf_{\P\in\Pi^*}\sup_{\pi\in\V}g(\pi;\P) \le\inf_{c\in\C
}
\lim_{n\to
\infty}\lambda^{*,c}_n =\lambda^*.
\]
\upqed\end{pf}
%
\begin{rem}
Note that the normalized eigenfunction $\eta^*$ in the statement of
Theorem~\ref{thmmain} may not be unique. It follows that the set of
measures $\Pi^*$ and the min--max problem in (\ref{min-maxproblem}) may
differ with our choice of $\eta^*$. In spite of this, we would like to
emphasize the following:
\begin{longlist}[(ii)]
\item[(i)] No matter which $\eta^*$ we choose, the robust maximal
asymptotic growth rate $\lambda^*$ stays the same.
\item[(ii)] At the first glance, it may seem restrictive to work with
$\Pi^*$. However, by the same calculation in~\cite{KR10}, Remark 2.2,
we see that: no matter which $\eta^*$ we choose, $\Pi^*$ is large
enough to contain all the probabilities in $\Pi$ under which $X$ is
tight in~$E$, and thus corresponds to those $\P\in\Pi$ such that $X$
is stable.
\end{longlist}
\end{rem}

\begin{appendix}\label{app}
\section{\texorpdfstring{Continuous selection results needed for Proposition \lowercase{\protect\ref{proplamb_d_inflambc_d}}}
{Continuous selection results needed for Proposition 3.3}}\label{secappendix}

The goal of this Appendix is to state and prove Proposition \ref
{propcmexists}, which is used in the proof of Proposition \ref
{proplamb_d_inflambc_d}. Before we do that, we need some preparations
concerning the theory of continuous selection in~\cite{Michael56} and
\cite{Brown89}.

\begin{defn} Let $X$ be a topological space.
\begin{enumerate}[(iii)]
\item[(i)] We say $X$ is a $T_1$ space if for any distinct points
$x,y\in X$, there exist open sets $U_x$ and $U_y$ such that $U_x$
contains $x$ but not $y$, and $U_y$ contains $y$ but not~$x$.
\item[(ii)] We say $X$ is a $T_2$(Hausdorff) space if for any distinct
points $x,y\in X$, there exist open sets $U_x$ and $U_y$ such that
$x\in U_x$, $y\in U_y$ and $U_x\cap U_y=\varnothing$.
\item[(iii)] We say $X$ is a paracompact space if for any collection
$\{X_\alpha\}_{\alpha\in\A}$ of open sets in $X$ such that $\bigcup_{\alpha\in\A}X_\alpha=X$, there exists a collection $\{X_\beta\}_{\beta
\in\B}$ of open sets in $X$ satisfying:
\begin{enumerate}[(3)]
\item[(1)] each $X_\beta$ is a subset of some $X_\alpha$;
\item[(2)] $\bigcup_{\beta\in\B}X_\beta=X$;
\item[(3)] given $x\in X$, there exists an open neighborhood of $x$
which intersects only finitely many elements in $\{X_\beta\}_{\beta
\in\B}$.
\end{enumerate}
\end{enumerate}
\end{defn}
%
\begin{defn}
Let $X, Y$ be topological spaces. A set-valued map $\phi\dvtx X\mapsto2^Y$
is lower semicontinuous if, whenever $V\subset Y$ is open in $Y$, the
set $\{x\in X \mid\phi(x)\cap V\neq\varnothing\}$ is open in $X$.
\end{defn}

The main theorem in~\cite{Michael56}, Theorem 3.2$''$, gives the
following result for continuous selection.
%
\begin{prop}
Let $X$ be a $T_1$ paracompact space, $Y$ be a Banach space and $\phi
\dvtx X\mapsto2^{Y}$ be a set-valued map such that $\phi(x)$ is a closed
convex subset of $Y$ for each $x\in X$. Then, if $\phi$ is lower
semicontinuous, there exists a continuous function $f\dvtx X\mapsto Y$ such
that $f(x)\in\phi(x)$ for all $x\in X$.
\end{prop}

Since the lower semicontinuity of $\phi$ can be difficult to prove in
general, one may wonder whether there is a weaker condition sufficient
for continuous selection. Brown~\cite{Brown89} worked toward this
direction and characterized the weakest possible condition (it is
therefore sufficient and necessary). For the special case where $X$ is
a Hausdorff paracompact space and $Y$ is a real linear space with
finite dimension~$n^*$, given a set-valued map $\phi\dvtx X\mapsto2^Y$, a
sequence $\{\phi^{(n)}\}_{n\in\N}$ of set-valued maps was introduced in
\cite{Brown89} via the following iteration:
%
\begin{eqnarray}
\label{phi^n} \phi^{(1)}(x)&:=& \bigl\{y\in\phi(x) \mid
\mbox{Given } V \mbox{ open in } Y \mbox{ s.t. } y\in V,\nonumber\\
&&\hspace*{3pt}\mbox{ there
is a neighborhood $U$ of $x$}
\mbox{ s.t. } \forall x'\in U, \exists
y'\in\phi \bigl(x'\bigr)\cap V\bigr\};\hspace*{-35pt}
\\
\phi^{(n)}(x)&:=&\bigl(\phi^{(n-1)}\bigr)^{(1)}(x)
\qquad\mbox{for } n\ge2.\nonumber
\end{eqnarray}
The following result, taken from~\cite{Brown89}, Theorem 4.3,
characterizes the possibility of continuous selection using $\phi^{(n^*)}$.
%
\begin{prop}\label{propphi^nnotempty}
Let $X$ be a Hausdorff paracompact space, $Y$ be a real linear space
with finite dimension $n^*$ and $\phi\dvtx X\mapsto2^{Y}$ be a set-valued
map such that $\phi(x)$ is a closed convex subset of $Y$ for each
$x\in
X$. Then, there exists a continuous function $f\dvtx X\mapsto Y$ such that
$f(x)\in\phi(x)$ for all $x\in X$ if and only if $\phi^{(n^*)}(x)\neq
\varnothing$ for all $x\in X$.
\end{prop}

In this paper, we would like to take $X= \bar{D}$ and $Y=\bS^d$, where
$D\subset E$ is a smooth bounded domain. Note that $\bar{D}$ is
Hausdorff and paracompact as it is a metric space in $\R^d$ (see, e.g.,
\cite{Kelley55}, Corollary 5.35), and $\bS^d$ is a real linear space
with dimension $n^*:=d(d+1)/2$. Fix two continuous functions $\gamma
,\Gamma\dvtx  E\mapsto(0,\infty)$ with $\gamma\le\Gamma$, we consider the
operator $F_{\gamma,\Gamma}\dvtx E\times\bS^d\mapsto\R$ defined by
%
\begin{equation}
\label{defnFgamma} F_{\gamma,\Gamma}(x,M):=\frac{1}{2}
\M^+_{\gamma(x),\Gamma
(x)}(M)=\frac
{1}{2}\sup_{A\in\A(\gamma(x),\Gamma(x))}\operatorname{Tr}(AM).
\end{equation}
Observe that $F_{\gamma,\Gamma}$ also satisfies (\ref{H1})--(\ref{H4}),
and in particular $F_{\theta,\Theta}=F$. Given $m\in\N$, we intend to
show that there exists a continuous function $c_m\dvtx \bar{D}\mapsto\bS^d$
such that for all $x\in\bar{D}$, $c_m(x)\in\A(\gamma(x),\Gamma(x))$
and $F_{\gamma,\Gamma}(x,\break D^2\eta_{D})\le L^{c_m(\cdot)}\eta_{D}(x)+1/m$, with $\eta_{D}$ given in Lemma
\ref{lemexistence}.
Note that since $\eta_{D}\in C^{2}(\bar{D})$ by Proposition~\ref
{propexistence}, $D^2\eta_{D}$ is well defined on $\partial D$.
Also, see Proposition~\ref{proplamb_d_inflambc_d} for the purpose of
finding such a function $c_m$. We then define the set-valued map
$\varphi\dvtx D\mapsto\bS^d$ by
%
\begin{eqnarray}
\varphi(x)&:=& \bigl\{M\in\bS^d \mid M\in\A\bigl(\gamma (x),
\Gamma(x)\bigr) \mbox{ and }\nonumber\\[-8pt]\\[-8pt]
&&\hspace*{5pt} F_{\gamma,\Gamma}\bigl(x,D^2
\eta_{D}\bigr)\le L^M\eta_{D}(x)+1/m \bigr\}.\nonumber
\end{eqnarray}
For any $x\in\bar{D}$, we see from the definition of $F_{\gamma
,\Gamma
}$ that $\varphi(x)\neq\varnothing$. Moreover, $\varphi(x)$ is by
definition a closed convex subset of $\bS^d$. Then, we define $\varphi^{(n)}$ inductively as in (\ref{phi^n}) for all $n\in\N$. In view of
Proposition~\ref{propphi^nnotempty}, such a function $c_m$ exists if
$\varphi^{(n^*)}(x)\neq\varnothing$ for all $x\in\bar{D}$. We claim that
this is true. Actually, we will prove a stronger result in the next
lemma: given $x\in\bar{D}$, $\varphi^{(n)}(x)\neq\varnothing$ for all
$n\in\N$.

Recall that $B_{\delta}(x)$ denotes the open ball in $\R^d$ centered at
$x\in\R^d$ with radius $\delta>0$. In the following, we will denote by
$B^{\bar{D}}_\delta(x)$ the corresponding open ball in $\bar{D}$ under
the relative topology, that is, $B^{\bar{D}}_\delta(x):=B_\delta
(x)\cap
\bar{D}$. Similarly, we will denote by $B^{\bS^d}_\delta(M)$ the
corresponding open ball in $\bS^d$ under the topology induced by \mbox{$\|
\cdot\|$} in (\ref{norm1}).
%
\begin{lem}\label{lemAinphi}
Fix a smooth bounded domain $D\subset E$, two continuous functions
$\gamma,\Gamma\dvtx E\mapsto(0,\infty)$ with $\gamma\le\Gamma$, and
$m\in\N
$. Let $\eta_{D}$ be given as in Lemma~\ref{lemexistence}. Then,
given $x\in\bar{D}$, if $M\in\varphi(x)$ satisfies
%
\begin{equation}
\label{condition} F_{\gamma,\Gamma}\bigl(x,D^2\eta_{D}
\bigr)< L^M\eta_{D}(x)+1/m,
\end{equation}
then $M\in\varphi^{(n)}(x)$ for all $n\in\N$.
\end{lem}
\begin{pf}
Fix $M\in\varphi(x)$ such that (\ref{condition}) holds. We will first
show that $M\in\varphi^{(1)}(x)$, and then complete the proof by an
induction argument. Take $0\le\zeta<1/m$ such that $F_{\gamma,\Gamma
}(x,D^2\eta_{D})=L^M\eta_{D}(x)+\zeta$. Set $\nu:=1/m-\zeta>0$.
Recall that $\eta_{D}\in C^2(\bar{D})$ from Proposition \ref
{propexistence}. By the continuity of the maps $x\mapsto F_{\gamma
,\Gamma}(x,D^2\eta_{D}(x))$ [thanks to (\ref{H4})] and $x\mapsto
L^M\eta_{D}(x)$, we can take $\delta_1>0$ small enough such that
the following holds for any $x'\in B^{\bar{D}}_{\delta_1}(x)$:
%
\begin{eqnarray}
\label{sup<L^A} F_{\gamma,\Gamma}\bigl(x',D^2
\eta_{D}\bigr)&<& F_{\gamma,\Gamma}\bigl(x,D^2
\eta_{D}\bigr)+\frac{\nu}{3}=L^M\eta_{D}(x)+
\zeta+\frac{\nu}{3}\nonumber\\[-8pt]\\[-8pt]
&<& L^M\eta_{D}\bigl(x'
\bigr)+\zeta+\frac{2\nu}{3}.\nonumber
\end{eqnarray}
Since\vspace*{1pt} $\|\cdot\|_{\max}\le\|\cdot\|$, the map
$(M,y)\mapsto L^M\eta_{D}(y)$ is continuous in $M$, uniformly in
$y\in\bar{D}$. It follows that there exists $\beta>0$ such that
%
\begin{equation}
\label{L^BL^A}\qquad \|N-M\|<\beta \quad\Rightarrow\quad
\bigl|L^N\eta_{D}(y)-L^M\eta_{D}(y)\bigr|<
\frac
{\nu}{3}\qquad\mbox{for all } y\in\bar{D}.
\end{equation}
Now, by the continuity of $\gamma$ and $\Gamma$ on $\bar{D}$, we can
take $\delta_2>0$ such that $\max\{|\gamma(x')-\gamma(x)|,|\Gamma
(x')-\Gamma(x)|\}<\beta$ for all $x'\in B^{\bar{D}}_{\delta_2}(x)$. For
each $x'\in B^{\bar{D}}_{\delta_2}(x)$, we pick $M'\in\bS^d$ satisfying
\[
e_i\bigl(M'\bigr)= \cases{ \gamma\bigl(x'
\bigr), &\quad if $e_i(M)< \gamma\bigl(x'\bigr)$,
\vspace*{1pt}\cr
e_i(M), &\quad if $e_i(M)\in\bigl[\gamma\bigl(x'
\bigr),\Gamma\bigl(x'\bigr)\bigr]$,
\vspace*{1pt}\cr
\Gamma\bigl(x'
\bigr), &\quad if $e_i(M)> \Gamma\bigl(x'\bigr)$. }
\]
By construction, $M'\in\A(\gamma(x'),\Gamma(x'))$ and $\|M'-M\|\le
\max
\{|\gamma(x')-\gamma(x)|$, $|\Gamma(x')-\Gamma(x)|\}<\beta$, which implies
%
\begin{equation}
\label{L^AL^A} \bigl|L^{M'}
\eta_{D}(y)-L^M\eta_{D}(y)\bigr|<\frac{\nu}{3}
\qquad\mbox{for all } y\in\bar{D}.
\end{equation}
Finally, set $U:=B^{\bar{D}}_{\delta}(x)$ with $\delta:=\delta_1\wedge
\delta_2$. Then by (\ref{sup<L^A}) and (\ref{L^AL^A}), for any
$x'\in
U$ there exists $M'\in B^{\bS^d}_\beta(M)$ such that $M'\in\A
(\gamma
(x'),\Gamma(x'))$ and
%
\begin{equation}
\label{sup<L^AA} F_{\gamma,\Gamma}
\bigl(x',D^2\eta_{D}\bigr)<L^{M'}
\eta_{D}\bigl(x'\bigr)+1/m,
\end{equation}
which shows that $M'\in\varphi(x')$. Given any open set $V$ in $\bS^d$
such that $M\in V$, since we may take $\beta>0$ in (\ref{L^BL^A})
small enough such that $B^{\bS^d}_\beta(M)\subset V$, we conclude that
$M'\in V$ also. It follows that $M\in\varphi^{(1)}(x)$.

Notice that what we have proved is the following result: for any $x\in
\bar{D}$, if $M\in\varphi(x)$ satisfies (\ref{condition}), then
$M\in
\varphi^{(1)}(x)$. Since $M'\in\varphi(x')$ satisfies (\ref{sup<L^AA}),
the above result immediately gives $M'\in\varphi^{(1)}(x')$. We then
obtain a stronger result: for any $x\in\bar{D}$, if $M\in\varphi(x)$
satisfies (\ref{condition}), then $M\in\varphi^{(2)}(x)$. But this
stronger result, when applied again to $M'\in\varphi(x')$ satisfying
(\ref{sup<L^AA}), gives $M'\in\varphi^{(2)}(x')$. We, therefore, obtain
that: for any $x\in\bar{D}$, if $M\in\varphi(x)$ satisfies (\ref
{condition}), then $M\in\varphi^{(3)}(x)$. We can then argue inductively
to conclude that $M\in\varphi^{(n)}(x)$ for all $n\in\N$.
\end{pf}
%
\begin{prop}\label{propcmexists}
Fix a smooth bounded domain $D\subset E$ and two continuous functions
$\gamma,\Gamma\dvtx E\mapsto(0,\infty)$ with $\gamma\le\Gamma$. Let
$\eta_{D}$ be given as in Lemma~\ref{lemexistence}. For any $m\in\N$,
there exists a continuous function $c_m\dvtx \bar{D}\mapsto\bS^d$ such that
\begin{eqnarray}
c_m(x)\in\A\bigl(\gamma(x),\Gamma(x)\bigr) \quad\mbox{and}\quad
F_{\gamma
,\Gamma
}\bigl(x,D^2\eta_{D}\bigr)\le
L^{c_m(\cdot)}\eta_{D}(x)+1/m\nonumber\\
&&\eqntext{\mbox{for all } x\in\bar{D}.}
\end{eqnarray}
\end{prop}
\begin{pf}
Fix $m\in\N$. As explained before Lemma~\ref{lemAinphi}, $\bar{D}$
is a Hausdorff paracompact space, $\bS^d$ is a real linear space with
dimension $n^*:=d(d+1)/2$ and $\varphi(x)$ is a closed convex subset of
$\bS^d$ for all $x\in\bar{D}$. For each $x\in\bar{D}$, by the
definition of $F_{\gamma,\Gamma}$ in (\ref{defnFgamma}), we can
always find some $M\in\varphi(x)$ satisfying (\ref{condition}). By
Lemma~\ref{lemAinphi}, this implies $\varphi^{(n)}(x)\neq
\varnothing$
for all $n\in\N$. In particular, we have $\varphi^{(n^*)}(x)\neq
\varnothing$ for all $x\in\bar{D}$. Then the desired result follows from
Proposition~\ref{propphi^nnotempty}.
\end{pf}

\section{\texorpdfstring{Proof of Lemma \lowercase{\protect\ref{lem_harnack}}\textup{(ii)}}
{Proof of Lemma 3.6(ii)}}
\label{sec5}\label{secappendixB}

\begin{pf*}{Proof of (\ref{Harnack2})}
Pick $x_0\in D$ and $R_0>0$ such that $B_{R_0}(x_0)\subset D$. For any
$0<R<R_0$, define
\[
v_n(x):=u_n (x_0+Rx ) \quad\mbox{and}\quad \bar
{H}(x,M):=H (x_0+Rx,M ).
\]
Then we deduce from (\ref{H1}) and (\ref{uneqn}) that
\begin{eqnarray*}
\bar{H}\bigl(x,D^2v_n(x)\bigr)+R^2
\delta_n v_n(x)&=&H \bigl(x_0+Rx,D^2v_n(x)
\bigr)+R^2\delta_n v_n(x)\\
&=&R^2f_n(x_0+Rx)
\qquad\mbox{in } B_{R_0/R}(0).
\end{eqnarray*}
Since $\bar{H}(x,M)$ satisfies (\ref{Hbdd}) in $B_{R_0/R}(0)$, we can
apply the estimate (\ref{Harnack}) to $v_n$ and get
\begin{eqnarray*}
\sup_{\bar{B}_R(x_0)}u_n &=& \sup_{\bar{B}_1(0)}v_n\le C
\Bigl\{\inf_{\bar
{B}_1(0)} v_n+R^2\|f_n
\|_{\mathcal{L}^d(B_{R_0}(x_0))} \Bigr\}\\
&=&C \Bigl\{ \inf_{\bar{B}_R(x_0)} u_n+R^2
\|f_n\|_{\mathcal
{L}^d(B_{R_0}(x_0))} \Bigr\},
\end{eqnarray*}
where $C>0$ depends only on $R_0$, $d$, $\lambda$, $\Lambda$, $\sup_n\delta_n$. 
\end{pf*}
\begin{pf*}{Proof of the H\"{o}lder continuity}
Now, fix a compact connected set $K\subset D$. Set $R_0:=\frac
{1}{2}d(\partial K, \partial D)>0$. By~\cite{McCoy65}, Lemma 2, there
exists some $k^*\in\N$ such that the set $K':=\{x\in\R^d\mid
d(x,K)\le
R_0\}\subset D$ has the following property: any two points in $K'$ can
be joined by a polygonal line of at most $k^*$ segments which lie
entirely in $K'$. Fix $x_0\in K'$. By the definition of $R_0$, we have
$B_{R_0}(x_0)\subset D$. For each $n\in\N$, we consider the
nondecreasing function $w^n\dvtx (0,R_0]\mapsto\R$ defined by
\[
w^n(R):= M^n_R-m^n_R
\qquad\mbox{where } M^n_R:=\max_{\bar
{B}_R(x_0)}u_n,
m^n_R:=\min_{\bar{B}_R(x_0)}u_n.\vadjust{\goodbreak}
\]
For each $R\in(0,R_0]$, we obtain from (\ref{uneqn}) that $\{
u_n-m^n_R\}_{n\in\N}$ is sequence of nonnegative continuous viscosity
solution to
\[
H\bigl(x,D^2\bigl(u_n-m^n_R
\bigr)\bigr)+\delta_n\bigl(u_n-m^n_R
\bigr)=f_n-\delta_n m^n_R
\qquad\mbox{in } B_R(x_0).
\]
By the estimate (\ref{Harnack2}), there is a constant $C>0$,
independent of $n$ and $x_0$, such that
%
\begin{eqnarray}
\label{est1}\quad M^n_{R/4}-m^n_R&=&
\sup_{\bar{B}_{R/4}(x_0)}\bigl(u_n(x)-m^n_R\bigr)
\le C\inf_{\bar
{B}_{R/4}(x_0)}\bigl(u_n(x)-m^n_R
\bigr)+A R^2\nonumber\\[-8pt]\\[-8pt]
&=& C\bigl(m^n_{R/4}-m^n_R
\bigr) +AR^2,\nonumber
\end{eqnarray}
where $A>0$ is a constant that depends on $C$ and $R_0$, but not $n$
[thanks to the uniform boundedness of $\{u_n\}_{n\in\N}$ and the
boundedness of $\{f_n\}_{n\in\N}$ in $\mathcal{L}^d(D)$]. Define
$\bar
{H}(x,M):= -H(x,-M)$. Then we deduce again from (\ref{uneqn}) that
$\{
M^n_R-u_n\}_{n\in\N}$ is a sequence of nonnegative continuous viscosity
solutions to
\begin{eqnarray*}
\bar{H}\bigl(x,D^2\bigl(M^n_R-u_n
\bigr)\bigr)+\delta_n\bigl(M^n_R-u_n
\bigr)&=&-H\bigl(x,D^2u_n\bigr)+ \delta_n
\bigl(M^n_R-u_n\bigr) \\
&=& -f_n+
\delta_n M^n_R \qquad\mbox{in }
B_R(x_0).
\end{eqnarray*}
Observe that $\bar{H}$ also satisfies (\ref{H1}) and (\ref{Hbdd}).
Thus, we can apply estimate (\ref{Harnack2}) and get
%
\begin{eqnarray}
\label{est2} M^n_{R}-m^n_{R/4}&=&
\sup_{\bar{B}_{R/4}(x_0)}\bigl(M^n_R-u_n(x)\bigr)
\le C\inf_{\bar
{B}_{R/4}(x_0)}\bigl(M^n_R-u_n(x)
\bigr)+A R^2\nonumber\\[-8pt]\\[-8pt]
&=& C\bigl(M^n_{R}-M^n_{R/4}
\bigr) +AR^2,\nonumber
\end{eqnarray}
where $C$ and $A$ are as above. Summing (\ref{est1}) and (\ref{est2}), we get
\[
w^n(R/4)=M^n_{R/4}-m^n_{R/4}
\le\frac
{C-1}{C+1}\bigl(M^n_R-m^n_R
\bigr)+A'R^2=\frac
{C-1}{C+1}w^n(R)+A'R^2,
\]
where $A'>0$ depends on $C$ and $R_0$, and is independent of $R$ and
$n$. By applying~\cite{GT-book-01}, Lemma 8.23, to the above inequality,
for any $\beta\in(0,1)$, we can find some $\tilde{C}>0$ (depending on
$C$, $R_0$ and $A'$, but not $n$) such that $w^n(R)\le\tilde
{C}R^\beta$, for all $R\le R_0$. This implies the following result: for
any $x,y\in K'$ with $|x-y|\le R_0$, we can take $x_0=x$ in the above
analysis and obtain $|u_n(x)-u_n(y)|\le w^n(|x-y|)\le\tilde
{C}|x-y|^\beta$ for all $n\in\N$. For the case where $|x-y|> R_0$,
recall that $x$ and $y$ can be joined by a polygonal line of $k$
segments which lie entirely in $K'$, for some $k\le k^*$. On the $j$th
segment, pick points $x^j_1, x^j_2,\ldots, x^j_{\ell_j}$ along the
segment such that $x^j_1, x^j_{\ell_j}$ are the two endpoints,
$|x^j_{i}-x^j_{i+1}|=R_0$ for $i=1,\ldots,\ell_j-2$ and $|x^j_{\ell
_j-1}-x^j_{\ell_j}|\le R_0$. Since $K'$ is bounded, there must be a
uniform bound $\ell^*>0$ such that $\ell_j\le\ell^*$ for all $j$. Then,
for all $n\in\N$, we have
\begin{eqnarray*}
\bigl|u_n(x)-u_n(y)\bigr|&\le&\sum_{j=1}^k
\sum_{i=1}^{\ell
_j-1}\bigl|u_n
\bigl(x^j_i\bigr)-u_n\bigl(x^j_{i+1}
\bigr)\bigr|\le\sum_{j=1}^k\sum
_{i=1}^{\ell
_j-1}\tilde{C}\bigl|x^j_i-x^j_{i+1}\bigr|^\beta
\\
&\le& k^*\ell^*\tilde{C}|x-y|^\beta.
\end{eqnarray*}
\upqed\end{pf*}
\end{appendix}

\section*{Acknowledgments}

We would like to express our gratitude to the anonymous Associate
Editor and referees whose comments helped us improve our paper significantly.



\printaddresses


\begin{thebibliography}{32}

\bibitem{ABG-book-11}
\begin{bbook}[mr]
\bauthor{\bsnm{Arapostathis},~\bfnm{Ari}\binits{A.}},
  \bauthor{\bsnm{Borkar},~\bfnm{Vivek~S.}\binits{V.~S.}} \AND
  \bauthor{\bsnm{Ghosh},~\bfnm{Mrinal~K.}\binits{M.~K.}}
(\byear{2012}).
\btitle{Ergodic Control of Diffusion Processes}.
\bseries{Encyclopedia of Mathematics and Its Applications}
\bvolume{143}.
\bpublisher{Cambridge Univ. Press}, \blocation{Cambridge}.
\bid{mr={2884272}}
\bptnote{check year}%
\bptok{imsref}%
\end{bbook}
\endbibitem

\bibitem{ALP95}
\begin{barticle}[auto:STB|2012/11/28|19:21:31]
\bauthor{\bsnm{Avellaneda},~\bfnm{M.}\binits{M.}},
  \bauthor{\bsnm{Levy},~\bfnm{A.}\binits{A.}} \AND
  \bauthor{\bsnm{Paras},~\bfnm{A.}\binits{A.}}
(\byear{1995}).
\btitle{Pricing and hedging derivative securities in markets with uncertain
  volatilities}.
\bjournal{Appl. Math. Finance}
\bvolume{2}
\bpages{73--88}.
\bptok{imsref}%
\end{barticle}
\endbibitem

\bibitem{BH11}
\begin{bmisc}[auto:STB|2012/11/28|19:21:31]
\bauthor{\bsnm{Bayraktar},~\bfnm{E.}\binits{E.}} \AND
  \bauthor{\bsnm{Huang},~\bfnm{Y.~J.}\binits{Y.~J.}}
(\byear{2011}).
\bhowpublished{On the multi-dimensional controller and stopper games. Technical
  report, Univ. Michigan. Available at \url{http://arxiv.org/abs/1009.0932}}.
\bptok{imsref}%
\end{bmisc}
\endbibitem

\bibitem{BD07}
\begin{barticle}[mr]
\bauthor{\bsnm{Birindelli},~\bfnm{I.}\binits{I.}} \AND
  \bauthor{\bsnm{Demengel},~\bfnm{F.}\binits{F.}}
(\byear{2007}).
\btitle{Eigenvalue, maximum principle and regularity for fully non linear
  homogeneous operators}.
\bjournal{Commun. Pure Appl. Anal.}
\bvolume{6}
\bpages{335--366}.
\bid{doi={10.3934/cpaa.2007.6.335}, issn={1534-0392}, mr={2289825}}
\bptok{imsref}%
\end{barticle}
\endbibitem

\bibitem{BD10}
\begin{barticle}[mr]
\bauthor{\bsnm{Birindelli},~\bfnm{I.}\binits{I.}} \AND
  \bauthor{\bsnm{Demengel},~\bfnm{F.}\binits{F.}}
(\byear{2010}).
\btitle{Eigenfunctions for singular fully nonlinear equations in unbounded
  domains}.
\bjournal{NoDEA Nonlinear Differential Equations Appl.}
\bvolume{17}
\bpages{697--714}.
\bid{doi={10.1007/s00030-010-0077-y}, issn={1021-9722}, mr={2740536}}
\bptok{imsref}%
\end{barticle}
\endbibitem

\bibitem{Borkar06-ICM}
\begin{bincollection}[mr]
\bauthor{\bsnm{Borkar},~\bfnm{Vivek~S.}\binits{V.~S.}}
(\byear{2006}).
\btitle{Ergodic control of diffusion processes}.
In \bbooktitle{International {C}ongress of {M}athematicians. {V}ol. {III}}
\bpages{1299--1309}.
\bpublisher{Eur. Math. Soc.}, \blocation{Z\"urich}.
\bid{mr={2275729}}
\bptok{imsref}%
\end{bincollection}
\endbibitem

\bibitem{Brown89}
\begin{barticle}[mr]
\bauthor{\bsnm{Brown},~\bfnm{A.~L.}\binits{A.~L.}}
(\byear{1989}).
\btitle{Set valued mappings, continuous selections, and metric projections}.
\bjournal{J.~Approx. Theory}
\bvolume{57}
\bpages{48--68}.
\bid{doi={10.1016/0021-9045(89)90083-X}, issn={0021-9045}, mr={0990803}}
\bptok{imsref}%
\end{barticle}
\endbibitem

\bibitem{BEQ05}
\begin{barticle}[mr]
\bauthor{\bsnm{Busca},~\bfnm{J{\'e}r{\^o}me}\binits{J.}},
  \bauthor{\bsnm{Esteban},~\bfnm{Maria~J.}\binits{M.~J.}} \AND
  \bauthor{\bsnm{Quaas},~\bfnm{Alexander}\binits{A.}}
(\byear{2005}).
\btitle{Nonlinear eigenvalues and bifurcation problems for {P}ucci's
  operators}.
\bjournal{Ann. Inst. H. Poincar\'e Anal. Non Lin\'eaire}
\bvolume{22}
\bpages{187--206}.
\bid{doi={10.1016/j.anihpc.2004.05.004}, issn={0294-1449}, mr={2124162}}
\bptok{imsref}%
\end{barticle}
\endbibitem

\bibitem{CC-book}
\begin{bbook}[mr]
\bauthor{\bsnm{Caffarelli},~\bfnm{Luis~A.}\binits{L.~A.}} \AND
  \bauthor{\bsnm{Cabr{\'e}},~\bfnm{Xavier}\binits{X.}}
(\byear{1995}).
\btitle{Fully Nonlinear Elliptic Equations}.
\bseries{American Mathematical Society Colloquium Publications}
\bvolume{43}.
\bpublisher{Amer. Math. Soc.}, \blocation{Providence, RI}.
\bid{mr={1351007}}
\bptok{imsref}%
\end{bbook}
\endbibitem

\bibitem{DHP10}
\begin{barticle}[mr]
\bauthor{\bsnm{Denis},~\bfnm{Laurent}\binits{L.}},
  \bauthor{\bsnm{Hu},~\bfnm{Mingshang}\binits{M.}} \AND
  \bauthor{\bsnm{Peng},~\bfnm{Shige}\binits{S.}}
(\byear{2011}).
\btitle{Function spaces and capacity related to a sublinear expectation:
  Application to {$G$}-{B}rownian motion paths}.
\bjournal{Potential Anal.}
\bvolume{34}
\bpages{139--161}.
\bid{doi={10.1007/s11118-010-9185-x}, issn={0926-2601}, mr={2754968}}
\bptnote{check year}%
\bptok{imsref}%
\end{barticle}
\endbibitem

\bibitem{DM06}
\begin{barticle}[mr]
\bauthor{\bsnm{Denis},~\bfnm{Laurent}\binits{L.}} \AND
  \bauthor{\bsnm{Martini},~\bfnm{Claude}\binits{C.}}
(\byear{2006}).
\btitle{A theoretical framework for the pricing of contingent claims in the
  presence of model uncertainty}.
\bjournal{Ann. Appl. Probab.}
\bvolume{16}
\bpages{827--852}.
\bid{doi={10.1214/105051606000000169}, issn={1050-5164}, mr={2244434}}
\bptok{imsref}%
\end{barticle}
\endbibitem

\bibitem{FK11}
\begin{barticle}[mr]
\bauthor{\bsnm{Fernholz},~\bfnm{Daniel}\binits{D.}} \AND
  \bauthor{\bsnm{Karatzas},~\bfnm{Ioannis}\binits{I.}}
(\byear{2011}).
\btitle{Optimal arbitrage under model uncertainty}.
\bjournal{Ann. Appl. Probab.}
\bvolume{21}
\bpages{2191--2225}.
\bid{doi={10.1214/10-AAP755}, issn={1050-5164}, mr={2895414}}
\bptok{imsref}%
\end{barticle}
\endbibitem

\bibitem{FS-book-06}
\begin{bbook}[mr]
\bauthor{\bsnm{Fleming},~\bfnm{Wendell~H.}\binits{W.~H.}} \AND
  \bauthor{\bsnm{Soner},~\bfnm{H.~Mete}\binits{H.~M.}}
(\byear{2006}).
\btitle{Controlled {M}arkov Processes and Viscosity Solutions},
\bedition{2nd} ed.
\bseries{Stochastic Modelling and Applied Probability}
\bvolume{25}.
\bpublisher{Springer}, \blocation{New York}.
\bid{mr={2179357}}
\bptok{imsref}%
\end{bbook}
\endbibitem

\bibitem{Fujisaki99}
\begin{barticle}[mr]
\bauthor{\bsnm{Fujisaki},~\bfnm{Masatoshi}\binits{M.}}
(\byear{1999}).
\btitle{On probabilistic approach to the eigenvalue problem for maximal
  elliptic operator}.
\bjournal{Osaka J. Math.}
\bvolume{36}
\bpages{981--992}.
\bid{issn={0030-6126}, mr={1745647}}
\bptok{imsref}%
\end{barticle}
\endbibitem

\bibitem{GT-book-01}
\begin{bbook}[mr]
\bauthor{\bsnm{Gilbarg},~\bfnm{David}\binits{D.}} \AND
  \bauthor{\bsnm{Trudinger},~\bfnm{Neil~S.}\binits{N.~S.}}
(\byear{2001}).
\btitle{Elliptic Partial Differential Equations of Second Order}.
\bpublisher{Springer}, \blocation{Berlin}.
\bnote{Reprint of the 1998 edition}.
\bid{mr={1814364}}
\bptok{imsref}%
\end{bbook}
\endbibitem

\bibitem{KR10}
\begin{barticle}[auto:STB|2012/11/28|19:21:31]
\bauthor{\bsnm{Kardaras},~\bfnm{C.}\binits{C.}} \AND
  \bauthor{\bsnm{Robertson},~\bfnm{S.}\binits{S.}}
(\byear{2012}).
\btitle{Robust maximization of asymptotic growth}.
\bjournal{Ann. Appl. Probab.}
\bvolume{22}
\bpages{1576--1610}.
\bptok{imsref}%
\end{barticle}
\endbibitem

\bibitem{KK07}
\begin{barticle}[mr]
\bauthor{\bsnm{Kawohl},~\bfnm{B.}\binits{B.}} \AND
  \bauthor{\bsnm{Kutev},~\bfnm{N.}\binits{N.}}
(\byear{2007}).
\btitle{Comparison principle for viscosity solutions of fully nonlinear,
  degenerate elliptic equations}.
\bjournal{Comm. Partial Differential Equations}
\bvolume{32}
\bpages{1209--1224}.
\bid{doi={10.1080/03605300601113043}, issn={0360-5302}, mr={2354491}}
\bptok{imsref}%
\end{barticle}
\endbibitem

\bibitem{Kelley55}
\begin{bbook}[mr]
\bauthor{\bsnm{Kelley},~\bfnm{John~L.}\binits{J.~L.}}
(\byear{1955}).
\btitle{General Topology}.
\bpublisher{Van Nostrand}, \blocation{New York}.
\bid{mr={0070144}}
\bptok{imsref}%
\end{bbook}
\endbibitem

\bibitem{LyonsTJ95}
\begin{barticle}[auto:STB|2012/11/28|19:21:31]
\bauthor{\bsnm{Lyons},~\bfnm{T.~J.}\binits{T.~J.}}
(\byear{1995}).
\btitle{Uncertain volatility and the risk-free synthesis of derivatives}.
\bjournal{Appl. Math. Finance}
\bvolume{2}
\bpages{117--133}.
\bptok{imsref}%
\end{barticle}
\endbibitem

\bibitem{McCoy65}
\begin{barticle}[mr]
\bauthor{\bsnm{McCoy},~\bfnm{J.~W.}\binits{J.~W.}}
(\byear{1965}).
\btitle{An extension of the concept of {$L\sb{n}$} sets}.
\bjournal{Proc. Amer. Math. Soc.}
\bvolume{16}
\bpages{177--180}.
\bid{issn={0002-9939}, mr={0173232}}
\bptok{imsref}%
\end{barticle}
\endbibitem

\bibitem{Meyer-book-00}
\begin{bbook}[mr]
\bauthor{\bsnm{Meyer},~\bfnm{Carl}\binits{C.}}
(\byear{2000}).
\btitle{Matrix Analysis and Applied Linear Algebra}.
\bpublisher{SIAM},
  \blocation{Philadelphia, PA}.
\bid{doi={10.1137/1.9780898719512}, mr={1777382}}
\bptok{imsref}%
\end{bbook}
\endbibitem

\bibitem{Michael56}
\begin{barticle}[mr]
\bauthor{\bsnm{Michael},~\bfnm{Ernest}\binits{E.}}
(\byear{1956}).
\btitle{Continuous selections. {I}}.
\bjournal{Ann. of Math. (2)}
\bvolume{63}
\bpages{361--382}.
\bid{issn={0003-486X}, mr={0077107}}
\bptok{imsref}%
\end{barticle}
\endbibitem

\bibitem{Nutz12-Quasi-Sure}
\begin{barticle}[mr]
\bauthor{\bsnm{Nutz},~\bfnm{Marcel}\binits{M.}}
(\byear{2012}).
\btitle{A quasi-sure approach to the control of non-{M}arkovian stochastic
  differential equations}.
\bjournal{Electron. J. Probab.}
\bvolume{17}
\bpages{23}.
\bid{doi={10.1214/EJP.v17-1892}, issn={1083-6489}, mr={2900464}}
\bptok{imsref}%
\end{barticle}
\endbibitem

\bibitem{NS10}
\begin{bmisc}[auto:STB|2012/11/28|19:21:31]
\bauthor{\bsnm{Nutz},~\bfnm{M.}\binits{M.}} \AND
  \bauthor{\bsnm{Soner},~\bfnm{H.~M.}\binits{H.~M.}}
(\byear{2010}).
\bhowpublished{Superhedging and dynamic risk measures under volatility
  uncertainty. Technical report, ETH Z\"{u}rich. Available at
  \url{http://arxiv.org/abs/1011.2958}.}
\bptok{imsref}%
\end{bmisc}
\endbibitem

\bibitem{Peng07}
\begin{bmisc}[mr]
\bauthor{\bsnm{Peng},~\bfnm{Shige}\binits{S.}}
(\byear{2007}).
\bhowpublished{$G$-Brownian motion and dynamic risk measure under volatility
uncertainty. Technical report, Shandong Univ.
Available at \texttt{\href{http://arxiv.org/abs/0711.2834}{http://arxiv.org/}
\href{http://arxiv.org/abs/0711.2834}{abs/0711.2834}}.}
\bptok{imsref}%
\end{bmisc}
\endbibitem

\bibitem{Pinsky-book}
\begin{bbook}[mr]
\bauthor{\bsnm{Pinsky},~\bfnm{Ross~G.}\binits{R.~G.}}
(\byear{1995}).
\btitle{Positive Harmonic Functions and Diffusion}.
\bseries{Cambridge Studies in Advanced Mathematics}
\bvolume{45}.
\bpublisher{Cambridge Univ. Press}, \blocation{Cambridge}.
\bid{doi={10.1017/CBO9780511526244}, mr={1326606}}
\bptok{imsref}%
\end{bbook}
\endbibitem

\bibitem{Pucci66b}
\begin{barticle}[mr]
\bauthor{\bsnm{Pucci},~\bfnm{Carlo}\binits{C.}}
(\byear{1966}).
\btitle{Maximum and minimum first eigenvalues for a class of elliptic
  operators}.
\bjournal{Proc. Amer. Math. Soc.}
\bvolume{17}
\bpages{788--795}.
\bid{issn={0002-9939}, mr={0199576}}
\bptok{imsref}%
\end{barticle}
\endbibitem

\bibitem{Pucci66}
\begin{barticle}[mr]
\bauthor{\bsnm{Pucci},~\bfnm{Carlo}\binits{C.}}
(\byear{1966}).
\btitle{Operatori ellittici estremanti}.
\bjournal{Ann. Mat. Pura Appl. (4)}
\bvolume{72}
\bpages{141--170}.
\bid{issn={0003-4622}, mr={0208150}}
\bptok{imsref}%
\end{barticle}
\endbibitem

\bibitem{QS08}
\begin{barticle}[mr]
\bauthor{\bsnm{Quaas},~\bfnm{Alexander}\binits{A.}} \AND
  \bauthor{\bsnm{Sirakov},~\bfnm{Boyan}\binits{B.}}
(\byear{2008}).
\btitle{Principal eigenvalues and the {D}irichlet problem for fully nonlinear
  elliptic operators}.
\bjournal{Adv. Math.}
\bvolume{218}
\bpages{105--135}.
\bid{doi={10.1016/j.aim.2007.12.002}, issn={0001-8708}, mr={2409410}}
\bptok{imsref}%
\end{barticle}
\endbibitem

\bibitem{Safonov88}
\begin{barticle}[mr]
\bauthor{\bsnm{Safonov},~\bfnm{M.~V.}\binits{M.~V.}}
(\byear{1988}).
\btitle{Classical solution of second-order nonlinear elliptic equations}.
\bjournal{Izv. Akad. Nauk SSSR Ser. Mat.}
\bvolume{52}
\bpages{1272--1287, 1328}.
\bid{issn={0373-2436}, mr={0984219}}
\bptok{imsref}%
\end{barticle}
\endbibitem

\bibitem{STZ11-aggregation}
\begin{barticle}[mr]
\bauthor{\bsnm{Soner},~\bfnm{H.~Mete}\binits{H.~M.}},
  \bauthor{\bsnm{Touzi},~\bfnm{Nizar}\binits{N.}} \AND
  \bauthor{\bsnm{Zhang},~\bfnm{Jianfeng}\binits{J.}}
(\byear{2011}).
\btitle{Quasi-sure stochastic analysis through aggregation}.
\bjournal{Electron. J. Probab.}
\bvolume{16}
\bpages{1844--1879}.
\bid{doi={10.1214/EJP.v16-950}, issn={1083-6489}, mr={2842089}}
\bptok{imsref}%
\end{barticle}
\endbibitem

\end{thebibliography}
\end{document}